\documentclass[
paper=a4,
pagesize,
DIV=14,
numbers=noendperiod,
captions=nooneline,
abstracton,
twoside
]{scrartcl}

\usepackage{stackrel}
\usepackage[T1]{fontenc}
\usepackage{lmodern}
\usepackage[american]{babel}
\usepackage{microtype}
\usepackage{tabularx}
\usepackage{dcolumn}
\usepackage{graphicx}
\usepackage{amsmath, amsthm}
\usepackage{graphicx}
\usepackage{subfig}

\usepackage[
hyperref,
table
]{xcolor}

\usepackage{scrpage2}
\usepackage{eso-pic}
\usepackage{rotating}
\usepackage{amsmath}
\usepackage{mathtools}
\usepackage{amssymb}
\usepackage{amsthm}
\usepackage{etoolbox}
\usepackage{bm}
\usepackage{bbm}
\usepackage{enumitem}
\usepackage{graphicx}
\usepackage{grffile}
\usepackage{tikz}
\usepackage{wrapfig}
\usepackage{tabularx}
\usepackage{siunitx}
\usepackage{booktabs}
\usepackage{multirow}
\usepackage{vruler}
\usepackage{fancyvrb}
\usepackage{listings}
\usepackage{csquotes}
\usepackage[
backend=bibtex,
style=numeric,
uniquename=init,
firstinits=true,
hyperref=true,
useprefix=true,
maxcitenames=2,
maxbibnames=6,
date=iso8601,
urldate=iso8601
]{biblatex}
\usepackage[
hypertexnames=false,
setpagesize=false,
pdfborder={0 0 0},
pdfstartview=Fit,
bookmarksopen=true,
bookmarksnumbered=true
]{hyperref}

\xdefinecolor{lightgray}{RGB}{247, 247, 247}
\xdefinecolor{semilightgray}{RGB}{240, 240, 240}
\xdefinecolor{middlegray}{RGB}{127, 127, 127}
\xdefinecolor{blue}{RGB}{58, 95, 205}
\xdefinecolor{deepskyblue}{RGB}{0, 154, 205}
\xdefinecolor{chocolate}{RGB}{205, 102, 29}

\clearscrheadings
\setkomafont{pageheadfoot}{\normalfont\normalcolor\sffamily}
\setkomafont{pagenumber}{\normalfont\normalcolor\sffamily}
\cehead{M.\ Hofert, A.\ Memartoluie, D.\ Saunders, T.\ Wirjanto}
\cohead{Improved Algorithms for Computing Worst VaR}
\setkomafont{captionlabel}{\normalfont\normalcolor\sffamily\bfseries}

\setlist{
  align=left,
  labelsep=*,
  leftmargin=*,
  topsep=1mm,
  itemsep=0mm
}
\newcommand*{\mysquare}{\rule[0.18em]{0.36em}{0.36em}}
\newcommand*{\mytriangle}{\raisebox{0.12em}{\resizebox{0.48em}{0.48em}{$\blacktriangleright$}}}
\newcommand*{\mybar}{\rule[0.32em]{0.62em}{0.08em}}
\setlist[itemize,1]{label={\mysquare\ }}
\setlist[itemize,2]{label={\mytriangle\ }}
\setlist[itemize,3]{label={\mybar\ }}
\setlist[enumerate,1]{label=\arabic*)}
\setlist[enumerate,2]{label=\arabic{enumi}.\arabic*)}
\setlist[enumerate,3]{label=\arabic{enumi}.\arabic{enumii}.\arabic*)}

\makeatletter
\newcommand\myisodate{\number\year-\ifcase\month\or 01\or 02\or 03\or 04\or 05\or 06\or 07\or 08\or 09\or 10\or 11\or 12\fi-\ifcase\day\or 01\or 02\or 03\or 04\or 05\or 06\or 07\or 08\or 09\or 10\or 11\or 12\or 13\or 14\or 15\or 16\or 17\or 18\or 19\or 20\or 21\or 22\or 23\or 24\or 25\or 26\or 27\or 28\or 29\or 30\or 31\fi}
\makeatother
\newcolumntype{d}[2]{D{.}{.}{#1.#2}}
\newcommand*{\abstractnoindent}{}
\let\abstractnoindent\abstract
\renewcommand*{\abstract}{\let\quotation\quote\let\endquotation\endquote
  \abstractnoindent}
\deffootnote[1em]{1em}{1em}{\textsuperscript{\thefootnotemark}}
\lstdefinestyle{input}{
  backgroundcolor=\color{semilightgray},
  commentstyle=\itshape\color{chocolate},
  keywordstyle=\color{blue},
  stringstyle=\color{deepskyblue},
  numbers=left,
  numberstyle=\color{middlegray}\tiny
}
\lstdefinestyle{output}{
  backgroundcolor=\color{lightgray}
}

\lstdefinestyle{Lstyle}{
  language=[LaTeX]TeX,
  texcs={},
  otherkeywords={}
}

\lstnewenvironment{Linput}[1][]{%
  \lstset{style=input, style=Lstyle}
  #1
}{\vspace{-0.25\baselineskip}}

\lstnewenvironment{Loutput}[1][]{%
  \lstset{style=output, style=Lstyle}
  #1
}{\vspace{-0.25\baselineskip}}

\lstdefinestyle{Rstyle}{
  language=R,
  literate={<-}{{$\bm\leftarrow$}}2{<<-}{{$\bm{\mathrel{\bm\leftarrow\mkern-14mu\leftarrow}$}}}2{<=}{{$\bm\le$}}2{>=}{{$\bm\ge$}}2{!=}{{$\bm\neq$}}2,
  keywords={if, else, repeat, while, function, for, in, next, break},
  otherkeywords={}
}

\lstnewenvironment{Sinput}[1][]{%
  \lstset{style=input, style=Rstyle}
  #1
}{\vspace{-0.25\baselineskip}}

\lstnewenvironment{Soutput}[1][]{%
  \lstset{style=output, style=Rstyle}
  #1
}{\vspace{-0.25\baselineskip}}

\DeclareNameAlias{sortname}{last-first}
\DefineBibliographyExtras{american}{\DeclareQuotePunctuation{}}
\renewbibmacro*{volume+number+eid}{
  \setunit*{\addcomma\space}%
  \printfield{volume}%
  \printfield{number}}
\DeclareFieldFormat*{number}{(#1)}
\DeclareFieldFormat*{title}{#1}
\DeclareFieldFormat{doi}{
  \ifhyperref
  {\href{http://dx.doi.org/#1}{\nolinkurl{doi:#1}}}
  {\nolinkurl{doi:#1}}}
\renewbibmacro*{in:}{}
\DeclareFieldFormat{isbn}{ISBN #1}
\DeclareFieldFormat{pages}{#1}
\DeclareFieldFormat{url}{\url{#1}}
\DeclareFieldFormat{urldate}{\mkbibparens{#1}}
\renewcommand*{\cite}[2][]{\textcite[#1]{#2}}

\newif\ifstarttheorem
\newtheoremstyle{mythmstyle}%
{0.5em}
{0.5em}
{}
{}
{\sffamily\bfseries\global\starttheoremtrue}
{}
{\newline}
{\thmname{#1}\ \thmnumber{#2}\ \thmnote{(#3)}}
\theoremstyle{mythmstyle}
\newtheorem{definition}{Definition}[section]
\newtheorem{proposition}[definition]{Proposition}
\newtheorem{lemma}[definition]{Lemma}

\newtheorem{example}[definition]{Example}
\newtheorem{algorithm}[definition]{Algorithm}

\makeatletter
\preto\itemize{%
  \if@inlabel
  \ifstarttheorem
  \mbox{}\par\nobreak\vskip\glueexpr-\parskip-\baselineskip+0.25em\relax\hrule\@height\z@
  \fi%
  \fi%
  \global\starttheoremfalse%
  \def\tempa{proof}%
  \ifx\tempa\mycurrenvir
  \ifstarttheorem
  \mbox{}\par\nobreak\vskip\glueexpr-\parskip-\baselineskip+0.25em\relax\hrule\@height\z@
  \fi%
  \fi%
  \global\starttheoremfalse%
}
\preto\enditemize{\global\starttheoremfalse}
\makeatother

\makeatletter
\preto\enumerate{%
  \if@inlabel
  \ifstarttheorem
  \mbox{}\par\nobreak\vskip\glueexpr-\parskip-\baselineskip+0.25em\relax\hrule\@height\z@
  \fi%
  \fi%
  \global\starttheoremfalse%
  \def\tempa{proof}%
  \ifx\tempa\mycurrenvir
  \ifstarttheorem
  \mbox{}\par\nobreak\vskip\glueexpr-\parskip-\baselineskip+0.25em\relax\hrule\@height\z@
  \fi%
  \fi%
  \global\starttheoremfalse%
}
\preto\endenumerate{\global\starttheoremfalse}
\makeatother

\newcommand*{\eps}{\varepsilon}

\newcommand*{\IN}{\mathbbm{N}}

\newcommand*{\IR}{\mathbbm{R}}

\newcommand*{\IP}{\mathbbm{P}}
\newcommand*{\IE}{\mathbbm{E}}

\newcommand*{\Par}{\operatorname{Par}}

\newcommand*{\VaR}{\operatorname{VaR}}
\newcommand*{\bVaR}{\operatorname{\overline{VaR}}}
\newcommand*{\ES}{\operatorname{ES}}

\newcommand*{\R}{\textsf{R}}

\begin{document}
\thispagestyle{plain}
\begin{center}
  \sffamily {\bfseries\LARGE Improved Algorithms for Computing Worst
    Value-at-Risk:\\ Numerical Challenges and the\\[2mm] Adaptive Rearrangement
    Algorithm}
  \par\bigskip
  {\Large Marius Hofert, Amir Memartoluie, David Saunders, Tony Wirjanto\par
  \bigskip\myisodate\par}
\end{center}
\begin{abstract}
  Numerical challenges inherent in algorithms for computing worst Value-at-Risk in
  homogeneous portfolios are identified and solutions as well as words of warning concerning their
  implementation are provided. Furthermore, both conceptual and computational
  improvements to the Rearrangement Algorithm for approximating worst
  Value-at-Risk for portfolios with arbitrary marginal loss distributions are
  given. In particular, a novel Adaptive Rearrangement Algorithm is introduced and
  investigated. These algorithms are implemented using the \R\ package \texttt{qrmtools}.
\end{abstract}
\minisec{Keywords}
Risk aggregation, model uncertainty, Value-at-Risk, Rearrangement Algorithm.
\minisec{MSC2010}
65C60, 62P05

\section{Introduction}
An integral part of Quantitative Risk Management is to analyze the
one-period ahead vector of losses $\bm{L}=(L_1,\dots,L_d)$, where $L_j$
represents the loss (a random variable) associated with a given business line or
risk type with counterparty $j$, $j\in\{1,\dots,d\}$, over a fixed time horizon.
For financial institutions, the \emph{aggregated loss}
\begin{align*}
  L^+=\sum_{j=1}^d L_j
\end{align*}
is of particular interest. Under Pillar I of the Basel Accords, financial
institutions are required to set aside capital to manage market, credit and
operational risk. To this end a risk measure $\rho(\cdot)$ is used to map the
aggregate position $L^+$ to $\rho(L^+) \in \IR$ for obtaining the amount of
capital required to account for future losses over a predetermined time period. As
a risk measure, Value-at-Risk ($\VaR_\alpha$) has been widely adopted by the
financial industry since the mid nineties. It is defined as the
$\alpha$-quantile of the distribution function $F_{L^+}$ of $L^+$, i.e.,
\begin{align*}
  \VaR_{\alpha}(L^+)=F^-_{L^+}(\alpha)=\inf\{x\in\IR: F_{L^+}(x)\ge \alpha\},
\end{align*}
where $F_{L^+}^-$ denotes the quantile function of $F_{L^+}$; see
\cite{embrechtshofert2013c} for more details. A well known drawback of
$\VaR_{\alpha}(L^+)$ as a risk measure is that
$\VaR_{\alpha}(L^+)$ is not necessarily subadditive unless $\bm{L}$ follows an
elliptical distribution; see, e.g., \cite{embrechtsfurrerkaufmann2009},
\cite[p.~241]{mcneilfreyembrechts2005}, \cite{embrechtspuccettirueschendorf2013}
and \cite{hofertmcneil2014b}.

There are various methods for estimating the marginal loss distributions
$F_1,\dots,F_d$ of $L_1,\dots,L_d$, respectively, but capturing the $d$-variate
dependence structure (i.e., the underlying copula $C$) of $\bm{L}$ is often more
difficult.  This is due to the fact that typically not much is known about $C$
and estimation often not feasible (e.g., for rare-event losses occurring in
different geographic regions). In this work we focus on the case
where $C$ is unknown; the case of partial information about $C$, is studied by
\cite{bernardrueschendorfvanduffel2013} and
\cite{bernarddenuitvanduffel2014}. In our case, one only knows that
$\VaR_\alpha(L^+)\in[\underline{\VaR}_{\alpha}(L^+),\ \bVaR_{\alpha}(L^+)]$
where $\underline{\VaR}_{\alpha}(L^+)$ and $\bVaR_{\alpha}(L^+)$ denote the best
and the worst $\VaR_\alpha(L^+)$ over all distribution functions of $\bm{L}$
with marginals $F_1,\dots,F_d$, respectively. Note that this interval can be
wide, but financial firms are interested in computing it (often in high
dimensions $d$) to determine their risk capital for $L^+$ within this range.  As
we show in this work, even for small $d$ (and other moderate parameter choices)
this can be challenging.

In particular, we investigate solutions in the homogeneous case (i.e.,
$F_1=\dots=F_d$) which are considered ``explicit''; see
\cite{embrechtspuccettirueschendorfwangbeleraj2014} (note that the
  formulas for both $\underline{\VaR}_{\alpha}(L^+)$ and $\bVaR_{\alpha}(L^+)$
  have typographical errors; we correct them below).
In the general, inhomogeneous case (i.e., not all $F_j$'s necessarily being
equal), we consider the Rearrangement Algorithm of
\cite{embrechtspuccettirueschendorf2013} for computing
$\underline{\VaR}_{\alpha}(L^+)$ and $\bVaR_{\alpha}(L^+)$. The presented
algorithms with conceptual and numerical improvements have been implemented in
the \R\ package \texttt{qrmtools}; see also the accompanying vignette
\texttt{VaR\_bounds} which provides further results, numerical investigations,
diagnostic checks and an application. All the results in this paper can be
reproduced with the package and vignette (and, obviously, other parameters can
be chosen if of interest). For a different approach for computing
$\underline{\VaR}_{\alpha}(L^+)$ and $\bVaR_{\alpha}(L^+)$ not discussed here,
see \cite{bernardmcleish2015}. In what follows, we focus on the worst
$\VaR_\alpha(L^+)$, i.e., $\bVaR_{\alpha}(L^+)$.

This paper is organized as follows. In Section~\ref{sec:known:sol} we highlight
and solve numerical challenges that practitioners may face when implementing
theoretical solutions for $\bVaR_{\alpha}(L^+)$ in the homogeneous case
$F_1=\dots=F_d$. Section~\ref{sec:RA} presents the main concept underlying the
Rearrangement Algorithm (RA) for computing $\underline{\VaR}_{\alpha}(L^+)$ and
$\bVaR_{\alpha}(L^+)$, levies criticism on its tuning parameters and
investigates its empirical performance using various test cases.
Section~\ref{sec:ARA} then presents a conceptually and numerically improved
version of the RA, which we call the Adaptive Rearrangement Algorithm (ARA), for
calculating $\underline{\VaR}_{\alpha}(L^+)$ and
$\bVaR_{\alpha}(L^+)$. Section~\ref{sec:con} concludes. Proofs of results
stated in the main body are relegated to an appendix.

\section{Known optimal solutions in the homogeneous case and their tractability}\label{sec:known:sol}
In order to assess the quality of general algorithms such as the RA, we need to
know (at least some) optimal solutions with which we can compare such
algorithms. \cite[Proposition~4]{embrechtspuccettirueschendorf2013} and
\cite[Proposition~1]{embrechtspuccettirueschendorfwangbeleraj2014} present
formulas for obtaining $\bVaR_\alpha(L^+)$ in
the homogeneous case. In this section, we address the corresponding numerical aspects and
algorithmic improvements. We assume $d\ge3$ throughout;
for $d=2$, \cite[Proposition~2]{embrechtspuccettirueschendorf2013} provide an
explicit solution for computing $\bVaR_\alpha(L^+)$ (under weak assumptions).

\subsection{Crude bounds for any $\VaR_\alpha(L^+)$}
The following lemma provides (crude) bounds for $\VaR_\alpha(L^+)$ which
are useful for computing initial intervals (see Section~\ref{sec:dual:bd}) and
conducting sanity checks. Note that we do not make any (moment or other)
assumptions on the involved marginal loss distributions; in fact, they do not
even have to be equal. Furthermore, the bounds do not depend on the underlying
unknown copula.
\begin{lemma}[Crude bounds for $\VaR_\alpha(L^+)$]\label{lem:crude}
  Let $L_j\sim F_j$, $j\in\{1,\dots,d\}$. For any $\alpha\in(0,1)$,
  \begin{align}
    d\min_j F_j^-(\alpha/d)\le\VaR_\alpha(L^+)\le d\max_j
    F_j^-\Bigl(1-\frac{1-\alpha}{d}\Bigr),\label{eq:crude:VaR:bounds}
  \end{align}
  where $F_j^-$ denotes the quantile function of $F_j$.
\end{lemma}
The bounds \eqref{eq:crude:VaR:bounds} can be computed with the function
\texttt{crude\_VaR\_bounds()} in the \R\ package \texttt{qrmtools}.

\subsection{The dual bound approach for computing $\bVaR_{\alpha}(L^+)$}\label{sec:dual:bd}
This approach for computing $\bVaR_{\alpha}(L^+)$ in the homogeneous case with
margin(s) $F$ is presented in
\cite[Proposition~4]{embrechtspuccettirueschendorf2013} and termed the
\emph{dual bound approach} in what follows; note that there is no corresponding
algorithm for computing $\underline{\VaR}_{\alpha}(L^+)$ with this approach. In
the remaining part of this subsection we assume that $F(0)=0$, $F(x)<1$ for all
$x\in[0,\infty)$ and that $F$ is absolutely continuous with ultimately
decreasing density. Let
\begin{align}
  D(s,t)=\frac{d}{s-dt}\int_t^{s-(d-1)t}\bar{F}(x)\,dx\quad\text{and}\quad D(s)=\min_{t\in[0,s/d]}D(s,t),\label{eq:D}
\end{align}
where $\bar{F}$ denotes the \emph{survival function of $F$}, i.e., $\bar{F}(x)=1-F(x)$.
In comparison to \cite[Proposition~4]{embrechtspuccettirueschendorf2013},
the \emph{dual bound} $D$ here uses a compact interval for $t$ (and thus
$\min\{\cdot\}$) by our requirement $F(0)=0$ and since $\lim_{t\uparrow
  s/d}D(s,t)=d\bar{F}(s/d)$ by l'Hospital's Rule. The procedure for computing
$\bVaR_{\alpha}(L^+)$ according to
\cite[Proposition~4]{embrechtspuccettirueschendorf2013} can now be given as
follows.
\begin{algorithm}[Computing $\bVaR_\alpha(L^+)$ according to the dual bound approach]\label{algo:dual}
  \begin{enumerate}
  \item Specify initial intervals $[s_l,s_u]$ and $[t_l,t_u]$.
  \item\label{algo:dual:inner} Inner root-finding in $t$: For each considered $s\in[s_l,s_u]$, compute $D(s)$ by
    iterating over $t\in[t_l,t_u]$ until a $t^*$ is found for which
    $h(s,t^*)=0$, where
    \begin{align*}
      h(s,t):=D(s,t)-(\bar{F}(t)+(d-1)\bar{F}(s-(d-1)t)).
    \end{align*}
    Then $D(s)=\bar{F}(t^*)+(d-1)\bar{F}(s-(d-1)t^*)$.
  \item Outer root-finding in $s$: Iterate Step~\ref{algo:dual:inner} over
    $s\in[s_l,s_u]$ until an $s^*$ is found for which $D(s^*)=1-\alpha$. Then
    return $s^*=\bVaR_{\alpha}(L^+)$.
  \end{enumerate}
\end{algorithm}
Algorithm~\ref{algo:dual} is implemented in the function
\texttt{worst\_VaR\_hom(..., method="dual")} in the \R\ package
\texttt{qrmtools}; the dual bound $D$ is available via
\texttt{dual\_bound()}. It requires a one-dimensional numerical
integration (unless $\bar{F}$ can be integrated
explicitly) within two nested calls of a root-finding algorithm
(\texttt{uniroot()} in \R). Note that Algorithm~\ref{algo:dual} requires specification of
the two initial intervals $[s_l,s_u]$ and $[t_l,t_u]$ and
\cite{embrechtspuccettirueschendorf2013} give no practical advice on how to
choose them.

First consider $[t_l,t_u]$. By our requirement $F(0)=0$ one can choose $t_l=0$
(or the infimum of the support of $F$). For $t_u$ one would like to choose
$s/d$; see the definition of $D(s)$ in \eqref{eq:D}. However, care has to be
taken as $h(s,s/d)=0$ for any $s$ and thus the inner root-finding procedure will
directly stop when a root is found at $t_u=s/d$. To take care of this, the inner
root-finding algorithm in \texttt{worst\_VaR\_hom(..., method="dual")} fixes
\texttt{f.upper}, i.e., $h(s,t_u)$ for the considered $s$, to $-h(s,0)$ so that
a root below $t_u=s/d$ can be detected; note that this is a (inelegant; for
lack of a better method) adjustment in the function value, and not in the
root-finding interval $[t_l,t_u]$.

Now consider $[s_l,s_u]$, in particular, $s_l$. According to
\cite[Proposition~4]{embrechtspuccettirueschendorf2013} $s_l$ has to be chosen
``sufficiently large''. If it is chosen too small, the inner root-finding
procedure in Step~\ref{algo:dual:inner} of Algorithm~\ref{algo:dual} will not be
able to locate a root; see also the left-hand side of
Figure~\ref{fig:dual}. There is currently no (good) solution known on how to
automatically determine a sufficiently large $s_l$ (note that
\texttt{worst\_VaR\_hom(..., method="dual")} currently requires $[s_l,s_u]$ to be
specified by the user). Given $s_l$, one can then choose $s_u$ as the maximum of
$s_l+1$ and the upper bound on $\VaR_\alpha(L^+)$ as given in
\eqref{eq:crude:VaR:bounds}, for example.

On the theoretical side, the following proposition implies that if $\bar{F}$ is
strictly convex, so is $D(s,\cdot)$ for fixed $s$. This shows the uniqueness of
the minimum when computing $D(s)$ as in \eqref{eq:D}. A standard result on
convexity of objective value functions then implies that $D(s)$ itself is also
convex; see \cite[Proposition~2.22]{rockafellarwets1998} and also the right-hand
side of Figure~\ref{fig:dual} below.
\begin{proposition}[Properties of $D(s,t)$ and $D(s)$]\label{prop:D}
  \begin{enumerate}
  \item $D(s)$ is decreasing.
  \item\label{prop:D:2} If $\bar{F}$ is convex, so is $D(s,t)$.
  \end{enumerate}
\end{proposition}

\begin{example}[Auxiliary functions for the dual bound approach]
  As an example, consider $d=8$ $\Par(\theta)$ risks with distribution function
  $F_j(x)=1-(1+x)^{-\theta_j}$, $x\ge0$, $\theta_j>0$. The left-hand side of
  Figure~\ref{fig:dual} illustrates $t\mapsto h(s,t)$ for $\theta=2$ and various
  $s$. Note that $h(s,s/d)$ is indeed $0$ and for (too) small $s$, $h(s,t)$ does
  not have a root for $t\in[0,s/d)$ as mentioned above. The right-hand side of
  Figure~\ref{fig:dual} shows the decreasing dual bound $D(s)$ for various
  parameters $\theta$.
  \begin{figure}[htbp]
    \centering
    \includegraphics[width=0.48\textwidth]{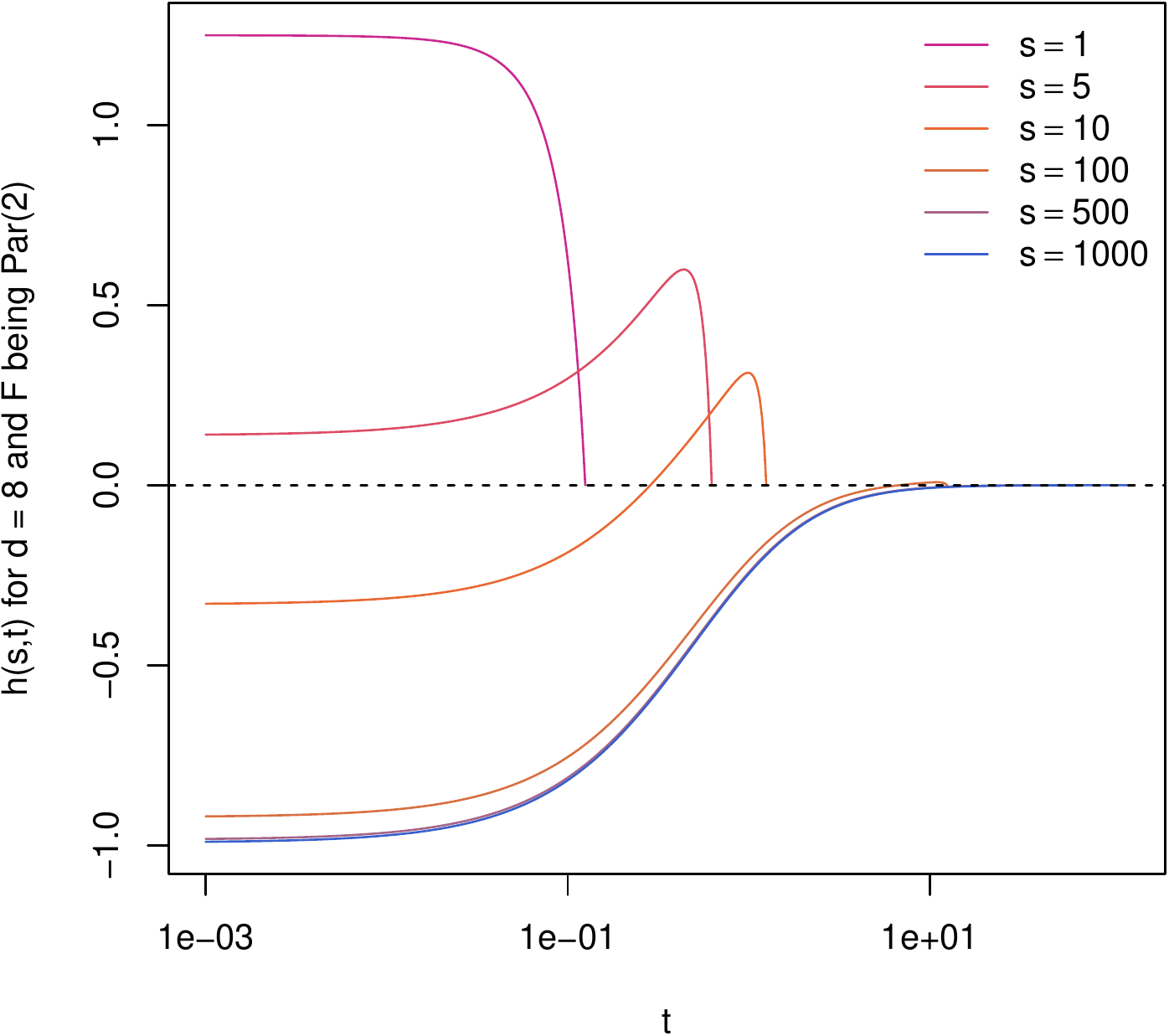}%
    \hfill
    \includegraphics[width=0.48\textwidth]{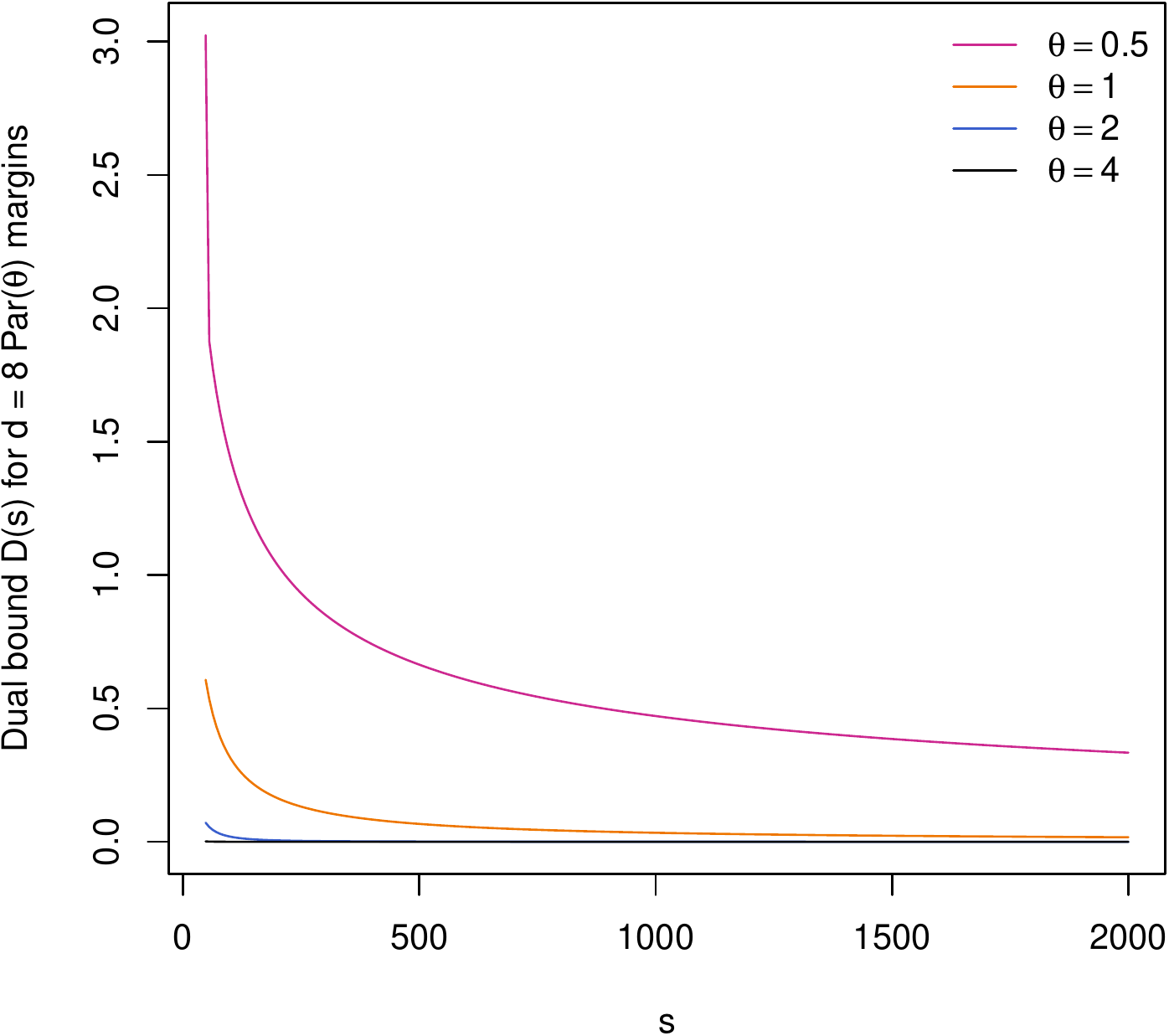}%
    \caption{$t\mapsto h(s,t)$ for various $s$, $d=8$ and $F$ being $\Par(2)$
      (left-hand side). The dual bound $D(s)$ for $d=8$ and $F$ being
      $\Par(\theta)$ for various parameters $\theta$ (right-hand side).}
    \label{fig:dual}
  \end{figure}
\end{example}

\subsection{Wang's approach for computing $\bVaR_\alpha(L^+)$}
\subsubsection*{Theory}
The approach mentioned in
\cite[Proposition~1]{embrechtspuccettirueschendorfwangbeleraj2014} is termed
\emph{Wang's approach} here. It originates from
\cite[Corollary~3.7]{wang2013bounds} and, thus, strictly speaking, precedes
the dual bound approach of
\cite[Proposition~4]{embrechtspuccettirueschendorf2013}.
It is conceptually simpler and numerically more stable than the dual bound
approach, yet still not straightforward to apply. For notational simplicity, let us introduce
\begin{align}
  a_c=\alpha+(d-1)c,\quad b_c=1-c,\label{eq:def:ab}
\end{align}
for $c\in[0, (1-\alpha)/d]$ (so that $a_c\in[\alpha,1-(1-\alpha)/d]$ and
$b_c\in[1-(1-\alpha)/d,1]$) and
\begin{align*}
  \bar{I}(c):=\frac{1}{b_c-a_c}\int_{a_c}^{b_c} F^-(y)\,dy,\quad c\in(0,
  (1-\alpha)/d], 
\end{align*}
with the assumption that $F$ admits a density which is positive and decreasing
on $[\beta,\infty)$ for some $\beta\le
F^-(\alpha)$. 
Then, for $L\sim F$,
\begin{align}
  \bVaR_\alpha(L^+)=d\,\IE[L\,|\,L\in[F^-(a_c),F^-(b_c)]],\quad\alpha\in[F(\beta),1),\label{eq:worst:VaR:Wang:cond:exp}
\end{align}
where $c$ (typically depending on $d$, $\alpha$) is the smallest number in
$(0,(1-\alpha)/d]$ such that
\begin{align}
  \bar{I}(c)\ge \frac{d-1}{d}F^-(a_c)+\frac{1}{d}F^-(b_c).\label{eq:worst:VaR:Wang:c}
\end{align}
In contrast to what is given in
\cite{embrechtspuccettirueschendorfwangbeleraj2014}, note that
$(0,(1-\alpha)/d]$ has to exclude 0 since otherwise, for $\Par(\theta)$ margins
with $\theta\in(0,1]$, $c$ equals 0 and thus, erroneously,
$\bVaR_\alpha(L^+)=\infty$. 
If $F$ is the distribution function of the $\Par(\theta)$ distribution,
$\bar{I}$ is given by
\begin{align}
  \bar{I}(c)=\begin{cases}
    \frac{1}{b_c-a_c}\frac{\theta}{1-\theta}((1-b_c)^{1-1/\theta}-(1-a_c)^{1-1/\theta})-1,&\text{if}\,\ \theta\neq1,\\
    \frac{1}{b_c-a_c}\log\bigl(\frac{1-a_c}{1-b_c}\bigr)-1,&\text{if}\,\ \theta=1.
  \end{cases}\label{eq:Ibar:Par}
\end{align}

The conditional distribution function $F_{L|L\in[F^-(a_c),F^-(b_c)]}$ of
$L\,|\,L\in[F^-(a_c),F^-(b_c)]$ is given by
$F_{L|L\in[F^-(a_c),F^-(b_c)]}(x)=\frac{F(x)-a_c}{b_c-a_c}$,
$x\in[F^-(a_c),F^-(b_c)]$.  Using this fact and a substitution, we obtain that,
for $\alpha\in[F(\beta),1)$, \eqref{eq:worst:VaR:Wang:cond:exp} becomes
\begin{align}
  \bVaR_\alpha(L^+)=d\int_{F^-(a_c)}^{F^-(b_c)}x\,dF_{L|L\in[F^-(a_c),F^-(b_c)]}(x)=d\,\frac{\int_{F^-(a_c)}^{F^-(b_c)}x\,dF(x)}{b_c-a_c}=d\bar{I}(c).\label{eq:worst:VaR:Wang}
\end{align}
Equation~\eqref{eq:worst:VaR:Wang} has the advantage of having the integration in
$\bar{I}(c)$ over a (at least theoretically) compact interval. Furthermore, finding the
smallest $c$ such that \eqref{eq:worst:VaR:Wang:c} holds also involves
$\bar{I}(c)$. A procedure thus only needs to know the quantile function $F^-$
to compute $\bVaR_\alpha(L^+)$. This leads to the following algorithm.
\begin{algorithm}[Computing $\bVaR_\alpha(L^+)$ according to Wang's approach]\label{alg:worst:VaR:Wang}
  \begin{enumerate}
  \item\label{algo:wang:1} Specify an initial interval $[c_l,c_u]$ with $0\le c_l<c_u<(1-\alpha)/d$.
  \item\label{algo:wang:2} Root-finding in $c$: Iterate over $c\in[c_l,c_u]$
    until a $c^*$ is found for which $h(c^*)=0$, where
    \begin{align}
      h(c):=\bar{I}(c)-\Bigl(\frac{d-1}{d}F^-(a_c)+\frac{1}{d}F^-(b_c)\Bigr),\quad
      c\in(0,(1-\alpha)/d].\label{eq:wang:h}
    \end{align}
    Then return $(d-1)F^-(a_{c^*})+F^-(b_{c^*})$.
  \end{enumerate}
\end{algorithm}
This procedure is implemented in the function \texttt{worst\_VaR\_hom(...,
  method="Wang")} in the \R\ package \texttt{qrmtools} with numerical
integration via \R's \texttt{integrate()} for computing $\bar{I}(c)$; the
function \texttt{worst\_VaR\_hom(...,
  method="Wang.Par")} makes use of \eqref{eq:Ibar:Par}.

The following proposition shows that the root-finding problem in
Step~\ref{algo:wang:2} of Algorithm~\ref{alg:worst:VaR:Wang} is well-defined in
the case of Pareto margins for all $\theta>0$ (including the infinite-mean
case); for other distributions under more restrictive assumptions, see \cite{bernardjiangwang2014}.
\begin{proposition}\label{prop:unique:root:Par}
  Let $F(x)=1-(1+x)^{-\theta}$, $\theta>0$, be the distribution function of the
  $\Par(\theta)$ distribution. Then $h$ in Step~\ref{algo:wang:2} of
  Algorithm~\ref{alg:worst:VaR:Wang} has a unique root on $(0, (1-\alpha)/d)$,
  for all $\alpha\in(0,1)$ and $d>2$.
\end{proposition}

\subsubsection*{Practice}
Let us now focus on the case of $\Par(\theta)$ margins (see
\texttt{worst\_VaR\_hom(..., method="Wang.Par")}) and, in particular, how to
choose the initial interval $[c_l,c_u]$ in
Algorithm~\ref{alg:worst:VaR:Wang}. We first consider $c_l$. $\bar{I}$ satisfies
$\bar{I}(0)=\frac{1}{1-\alpha}\int_{\alpha}^1F^-(y)\,dy=\ES_\alpha(L)$, i.e.,
$\bar{I}(0)$ is the \emph{expected shortfall} of $L\sim F$ at confidence level
$\alpha$. If $L$ has a finite first moment, then $\bar{I}(0)$ is
finite. Therefore, $h(0)$ is finite (if $F^-(1)<\infty$) or $-\infty$ (if
$F^-(1)=\infty$). Either way, one can take $c_l=0$. However, if $L\sim F$ has an
infinite first moment (see, e.g., \cite{hofertwuethrich2012} or
\cite{chavezdemoulinembrechtshofert2014} for situations in which this can
happen), then $\bar{I}(0)=\infty$ and $F^-(1)=\infty$, so $h(0)$ is not defined;
this happens, e.g., if $F$ is $\Par(\theta)$ with $\theta\in(0,1]$. In such a
case, we are forced to choose $c_l\in(0,(1-\alpha)/d)$; see the following
proposition for how this can be done \emph{theoretically}. Concerning $c_u$,
note that l'Hospital's Rule implies that $\bar{I}(c_u)=F^-(1-(1-\alpha)/d)$ and
thus that $h((1-\alpha)/d)=0$. We thus have a similar problem (a root at the
upper endpoint of the initial interval) as for computing the dual
bound. However, here we can construct a suitable $c_u<(1-\alpha)/d$; see the
following proposition.
\begin{proposition}[Computing $c_l,c_u$ for $\Par(\theta)$ risks]\label{prop:cl:cu:Par}
  Let $F$ be the distribution function of a $\Par(\theta)$ distribution,
  $\theta>0$. Then $c_l$ and $c_u$ in Algorithm~\ref{alg:worst:VaR:Wang}
  can be chosen as
  \begin{align*}
    c_l=\begin{cases}
      \frac{(1-\theta)(1-\alpha)}{d},&\text{if}\,\ \theta\in(0,1),\\
      \frac{1-\alpha}{(d+1)^{\frac{e}{e-1}}+d-1},&\text{if}\,\ \theta=1,\\
      \frac{1-\alpha}{(d/(\theta-1)+1)^{\theta}+d-1},&\text{if}\,\ \theta\in(1,\infty),
    \end{cases}\quad
    c_u=\begin{cases}
      \frac{(1-\alpha)(d-1+\theta)}{(d-1)(2\theta+d)},&\text{if}\,\ \theta\neq1,\\
      \frac{1-\alpha}{3d/2-1},&\text{if}\,\ \theta=1.
    \end{cases}
  \end{align*}
\end{proposition}

In the following examples we briefly consider some experiments which have led to
several numerical hurdles we had to overcome when implementing
\texttt{worst\_VaR\_hom(..., method="Wang.Par")}; see also
the detailed vignette \texttt{VaR\_bounds} for how to reproduce them (e.g., Section 1.4).
\begin{example}[$h$, $\underline{\VaR}_\alpha(L^+)$ and $\bVaR_\alpha(L^+)$ for Wang's approach
  with $\Par(\theta)$ risks]\label{ex:hom:h:Wang:Par}
  As an example, consider $\Par(\theta)$ risks and the confidence level
  $\alpha=0.99$. Figure~\ref{fig:worst:VaR:hom:Wang:h} illustrates the objective
  function $h(c)$ as a function of $c\in(0,(1-\alpha)/d]$ (see
  \eqref{eq:wang:h}) for various $\theta$ and $d=8$ (left-hand side) and $d=100$
  (right-hand side). Non-positive values $h(c)$ have been omitted so that the
  y-axis could be given in log-scale; note how steep $h$ can be (we have
  evaluated $h$ at $2^{13}+1$ points equally spaced between 0 and
  $(1-\alpha)/d$), especially for small $\theta$ (and large $d$). For an even
  larger number of evaluation points $c$, one sees that the first and last
  positive values of $h$ indeed reach down towards 0.  Overall, the objective
  function can be evaluated without numerical problems on $(0,(1-\alpha)/d]$ for
  our chosen $\theta$ and $d$.

  Figure~\ref{fig:worst:VaR:hom:Wang} displays $\underline{\VaR}_\alpha(L^+)$
  and $\bVaR_\alpha(L^+)$ as a function in $1-\alpha$ for various $\theta$ and
  $d=8$ (left-hand side) and $d=100$ (right-hand side).
  \begin{figure}[htbp]
    \centering
    \includegraphics[width=0.48\textwidth]{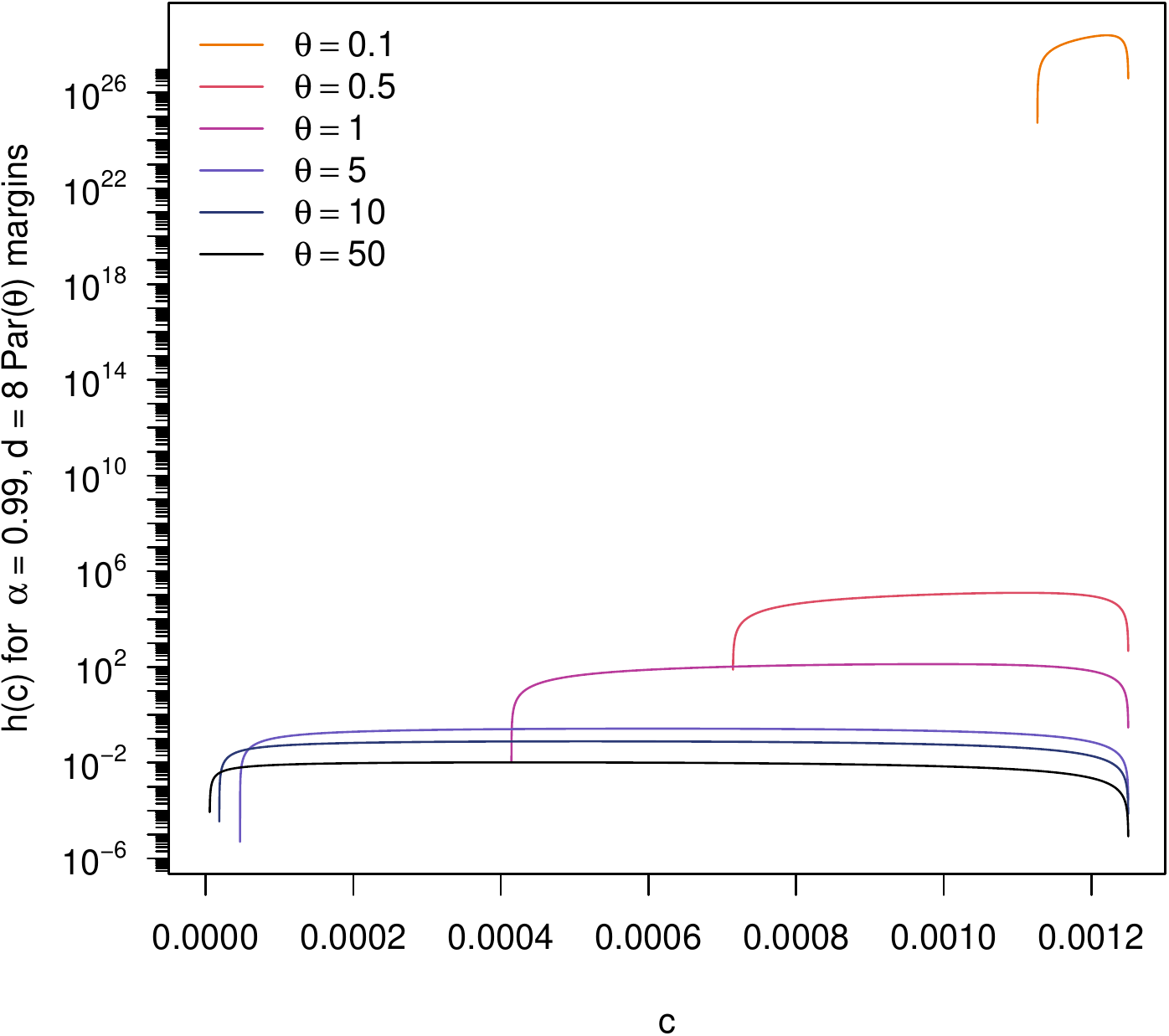}%
    \hfill
    \includegraphics[width=0.48\textwidth]{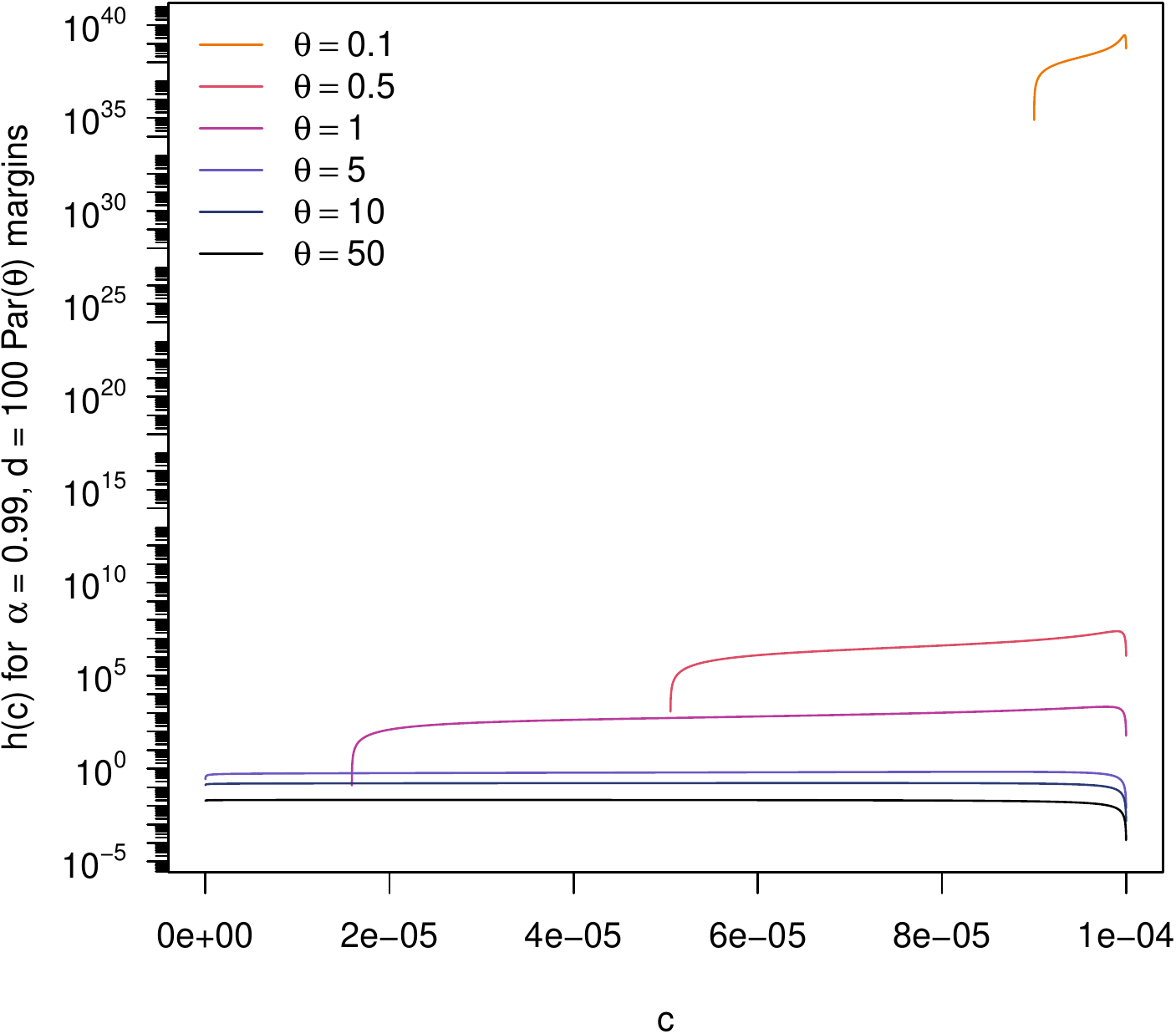}%
    \caption{Objective function $h(c)$ for $\alpha=0.99$, $F$ being
      $\Par(\theta)$, $d=8$ (left-hand side) and
      $d=100$ (right-hand side).}
    \label{fig:worst:VaR:hom:Wang:h}
  \end{figure}
  \begin{figure}[htbp]
    \centering
    \includegraphics[width=0.48\textwidth]{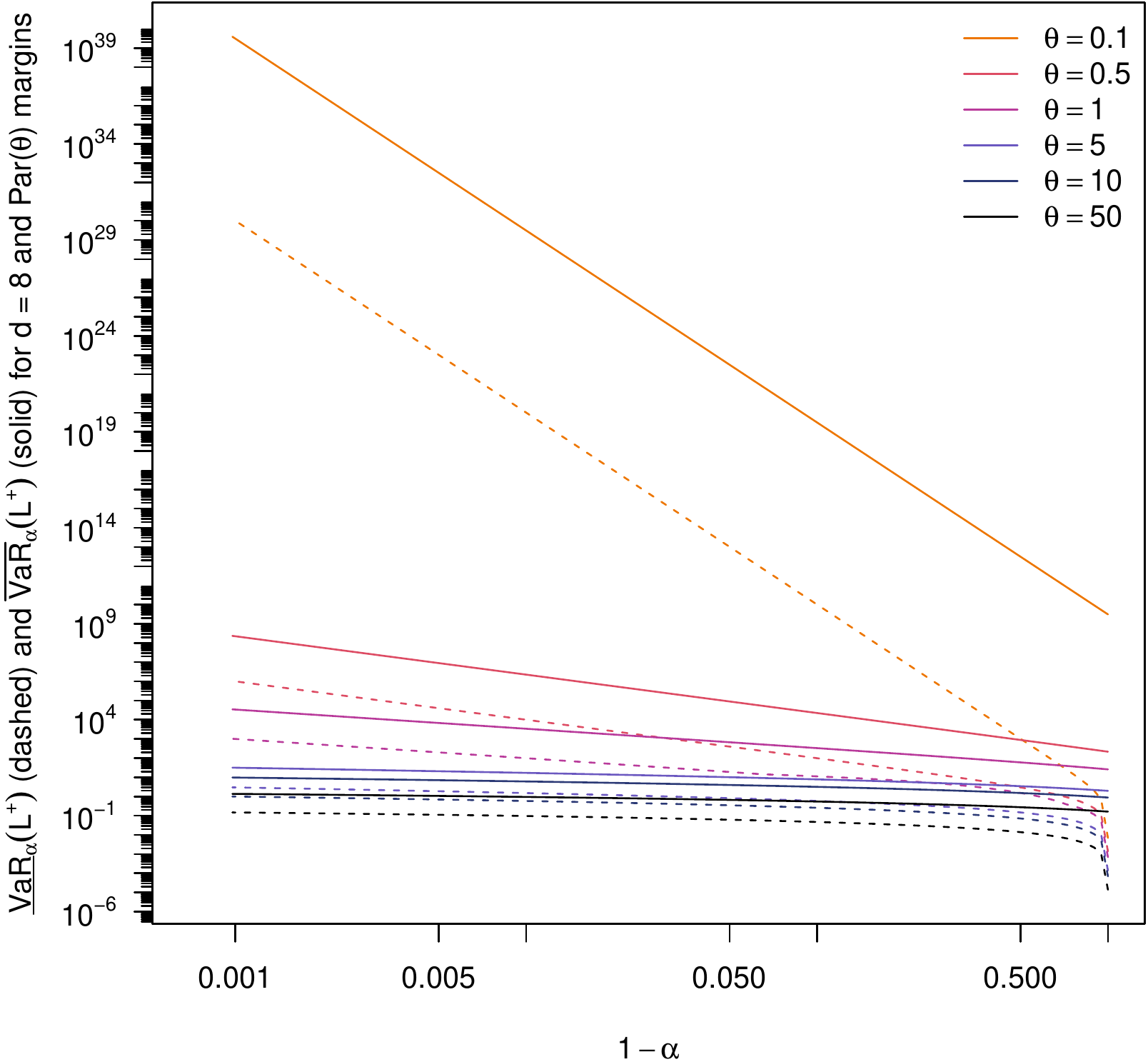}%
    \hfill
    \includegraphics[width=0.48\textwidth]{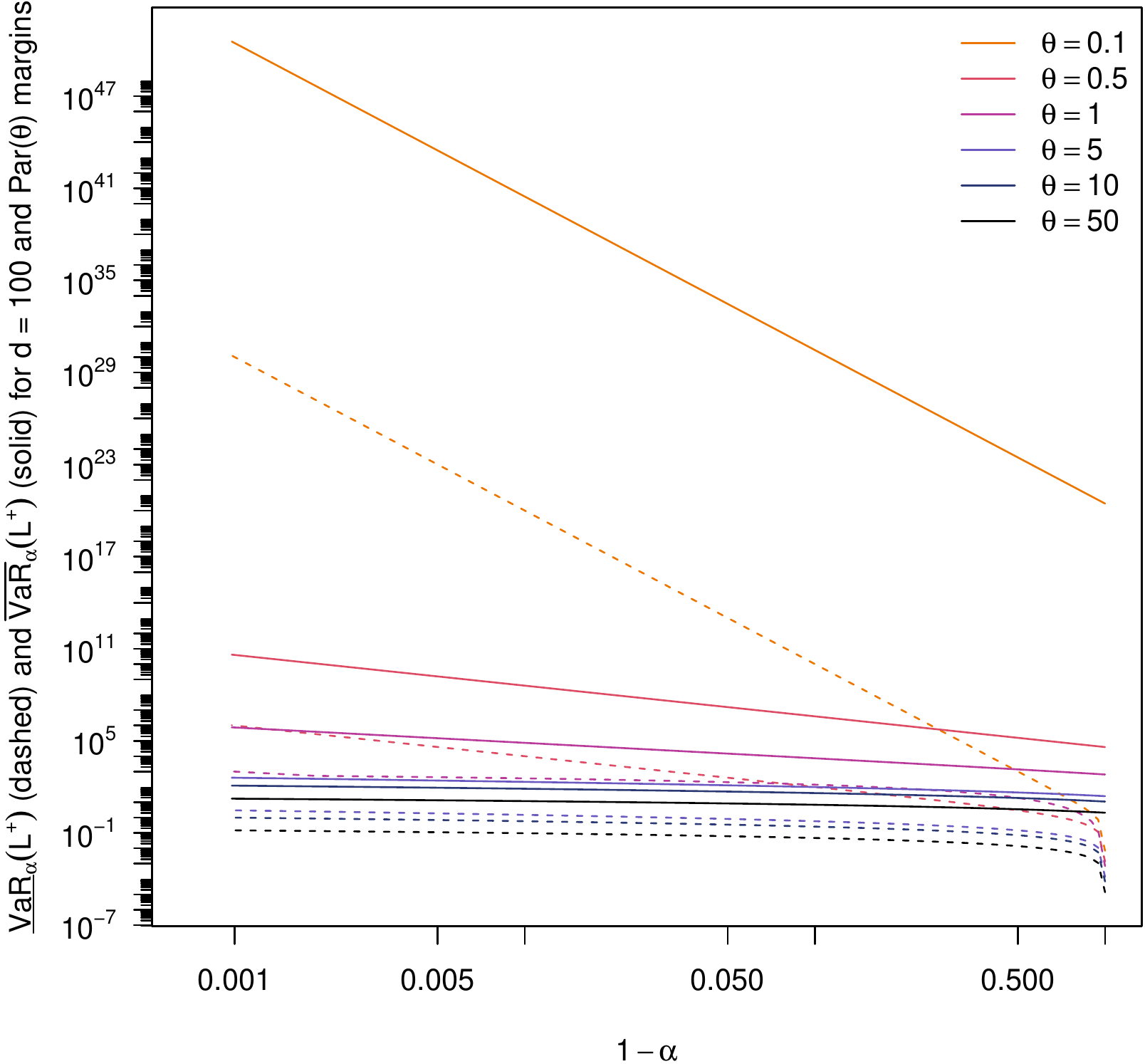}%
    \caption{$\underline{\VaR}_\alpha(L^+)$ and $\bVaR_\alpha(L^+)$ as functions of $1-\alpha$ for $F$ being $\Par(\theta)$, $d=8$ (left-hand side) and $d=100$ (right-hand side).}
    \label{fig:worst:VaR:hom:Wang}
  \end{figure}

  For obtaining numerically reliable results (over these wide ranges of
  parameters; indeed we tested much higher dimensions as well), one has to be
  careful when computing the root of $h$ for $c\in(0,(1-\alpha)/d)$. First,
  choosing a smaller root-finding tolerance is
  crucial. Figure~\ref{fig:hom:comparison} below (see also
  Example~\ref{ex:hom:compare}) shows what silently happens if this is not
  considered (our procedure chooses MATLAB's default $2.2204\cdot 10^{-16}$
  instead of the much larger \texttt{uniroot()} default
  \texttt{.Machine\$double.eps\^{}0.25}). Second, it turned out to be required
  to adjust the \emph{theoretically valid} initial interval described in
  Proposition~\ref{prop:cl:cu:Par} further in order to guarantee that $h$ is
  \emph{numerically} of opposite sign at the interval end points. In particular,
  \texttt{worst\_VaR\_hom(..., method="Wang.Par")} chooses $c_l/2$ as lower end
  point (with $c_l$ as in Proposition~\ref{prop:cl:cu:Par}) in case
  $\theta\neq 1$.

  These problems are described in detail in Section 1.4 of the vignette
  \texttt{VaR\_bounds}, where we also show that transforming the auxiliary
  function $h$ to a root-finding problem on $(1,\infty)$ as described in the
  proof of Proposition~\ref{prop:cl:cu:Par} does not require a smaller
  root-finding tolerance but also an extended initial interval and, furthermore,
  faces a cancellation problem (which can be solved, though); see also the
  left-hand side of Figure~\ref{fig:worst:VaR:hom:Wang:Par:numerics}, where we
  compare this approach to \texttt{worst\_VaR\_hom(..., method="Wang.Par")}
  \emph{after} fixing these numerical issues.

  In short, one should be very careful when implementing supposedly ``explicit
  solutions'' for computing $\underline{\VaR}_\alpha(L^+)$ or
  $\bVaR_\alpha(L^+)$ in the homogeneous case with $\Par(\theta)$ (and most
  likely also other) margins.
\end{example}

\begin{example}[Comparison of the approaches for $\Par(\theta)$ risks]\label{ex:hom:compare}
  Again let us consider $\Par(\theta)$ risks and the confidence level
  $\alpha=0.99$. Figure~\ref{fig:hom:comparison} compares Wang's approach (using
  numerical integration; see \texttt{worst\_VaR\_hom(..., method="Wang")}), Wang's
  approach (with an analytical formula for the integral $\bar{I}(c)$ but
  \texttt{uniroot()}'s default tolerance; see the vignette \texttt{VaR\_bounds}),
  Wang's approach (with an analytical formula for the integral $\bar{I}(c)$ and
  auxiliary function $h$ transformed to $(1,\infty)$; see the vignette
  \texttt{VaR\_bounds}), Wang's approach (with analytical formula for the integral
  $\bar{I}(c)$, smaller \texttt{uniroot()} tolerance and adjusted initial
  interval; see \texttt{worst\_VaR\_hom(..., method="Wang.Par")}), and the lower
  and upper bounds obtained from the RA (with absolute tolerance 0); see
  Section~\ref{sec:RA} and \texttt{RA()}. All of the results are divided by the
  values obtained from the dual bound approach to facilitate comparison. The two
  plots (for $d=8$ and $d=100$, respectively) show that comparable results are
  obtained by the different approaches and why it is important to use a smaller
  tolerance for \texttt{uniroot()} in Wang's approach.
  \begin{figure}[htbp]
    \centering
    \includegraphics[width=0.48\textwidth]{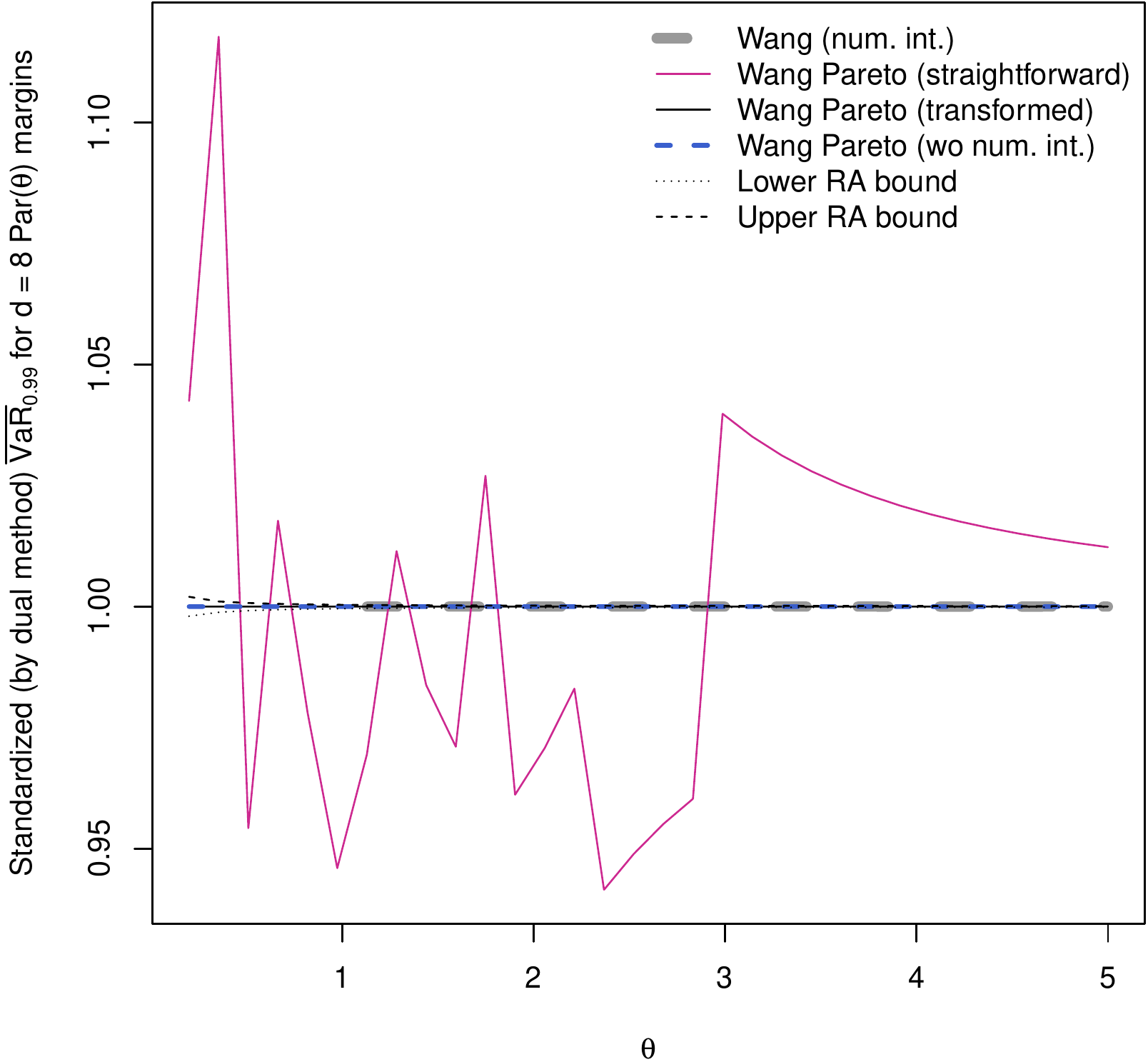}%
    \hfill
    \includegraphics[width=0.48\textwidth]{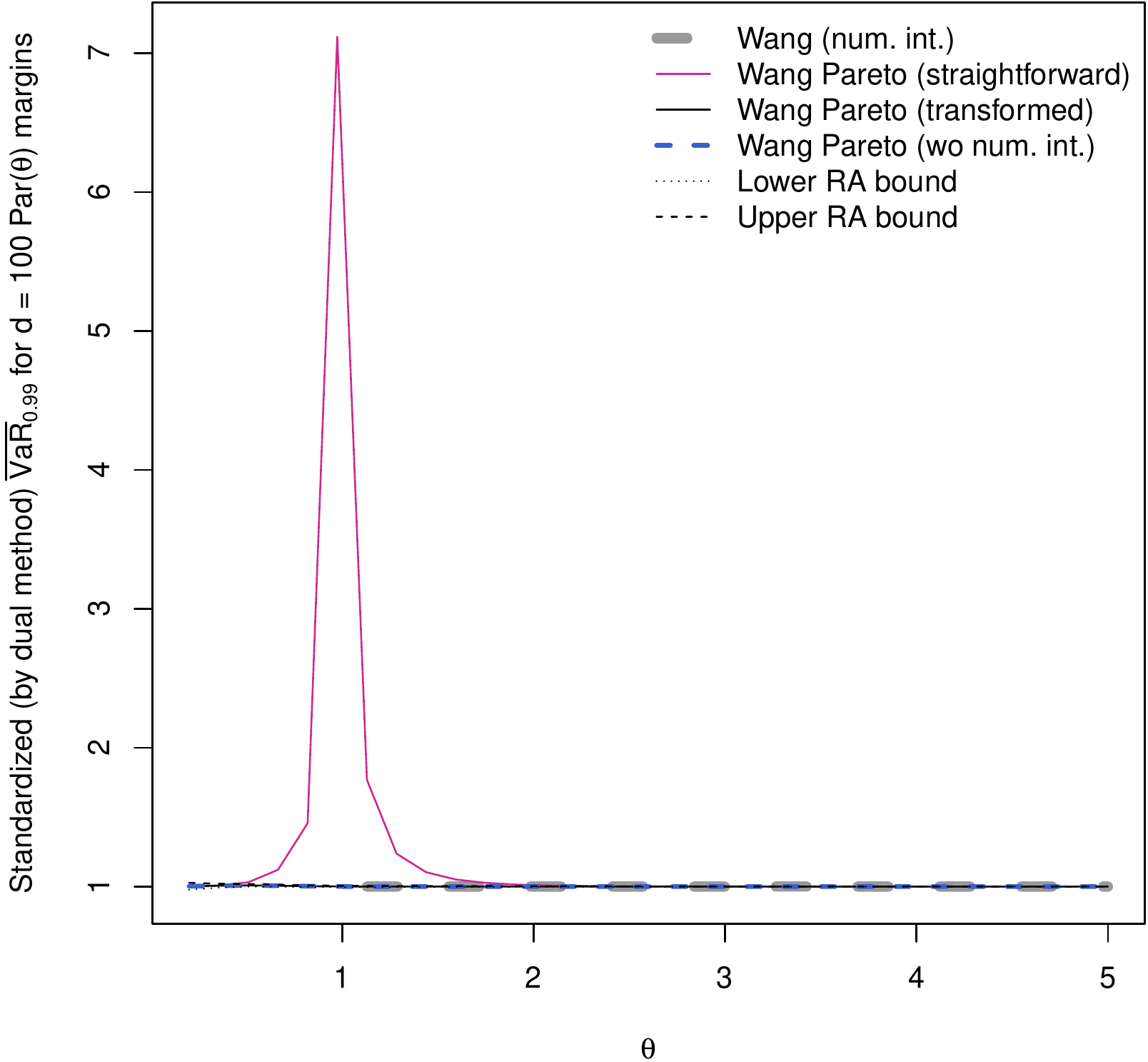}%
    \caption{Comparisons of Wang's approach (using numerical
      integration; see \texttt{worst\_VaR\_hom(..., method="Wang")}),
      Wang's approach (with an analytical formula for the integral
      $\bar{I}(c)$ but \texttt{uniroot()}'s default tolerance; see the vignette \texttt{VaR\_bounds}),
      Wang's approach (with an analytical formula for the integral $\bar{I}(c)$
      and auxiliary function $h$ transformed to $(1,\infty)$; see the vignette \texttt{VaR\_bounds}), 
      Wang's approach (with analytical formula for the integral
      $\bar{I}(c)$, smaller \texttt{uniroot()} tolerance and adjusted initial interval; see
      \texttt{worst\_VaR\_hom(..., method="Wang.Par")}),
      and the lower and upper bounds obtained from the RA;
      all of the results are divided by the values obtained from
      the dual bound approach to facilitate comparison.
      The left-hand side shows the case $d=8$, the right-hand side $d=100$.}
    \label{fig:hom:comparison}
  \end{figure}
\end{example}

Let us again stress how important the initial interval $[c_l,c_u]$ is. One could
be tempted to simply choose $c_u=(1-\alpha)/d$ and force the auxiliary function
$h$ to be of opposite sign at $c_u$, e.g., by setting $h(c_u)$ to
\texttt{.Machine\$double.xmin}, a positive but small
number. Figure~\ref{fig:worst:VaR:hom:Wang:wrong:right:adjustment} shows graphs
similar to the left-hand sides of Figures~\ref{fig:worst:VaR:hom:Wang} (but
$\bVaR_\alpha(L^+)$ only) and \ref{fig:hom:comparison} (standardized with
respect to the upper bound obtained from the RA). In particular,
$\bVaR_\alpha(L^+)$ is not monotone in $\alpha$ anymore (see the left-hand side
of Figure~\ref{fig:worst:VaR:hom:Wang:wrong:right:adjustment}) and the computed
$\bVaR_\alpha(L^+)$ values are not correct anymore (see the right-hand side of
Figure~\ref{fig:worst:VaR:hom:Wang:wrong:right:adjustment}).
\begin{figure}[htbp]
  \centering
  \includegraphics[width=0.48\textwidth]{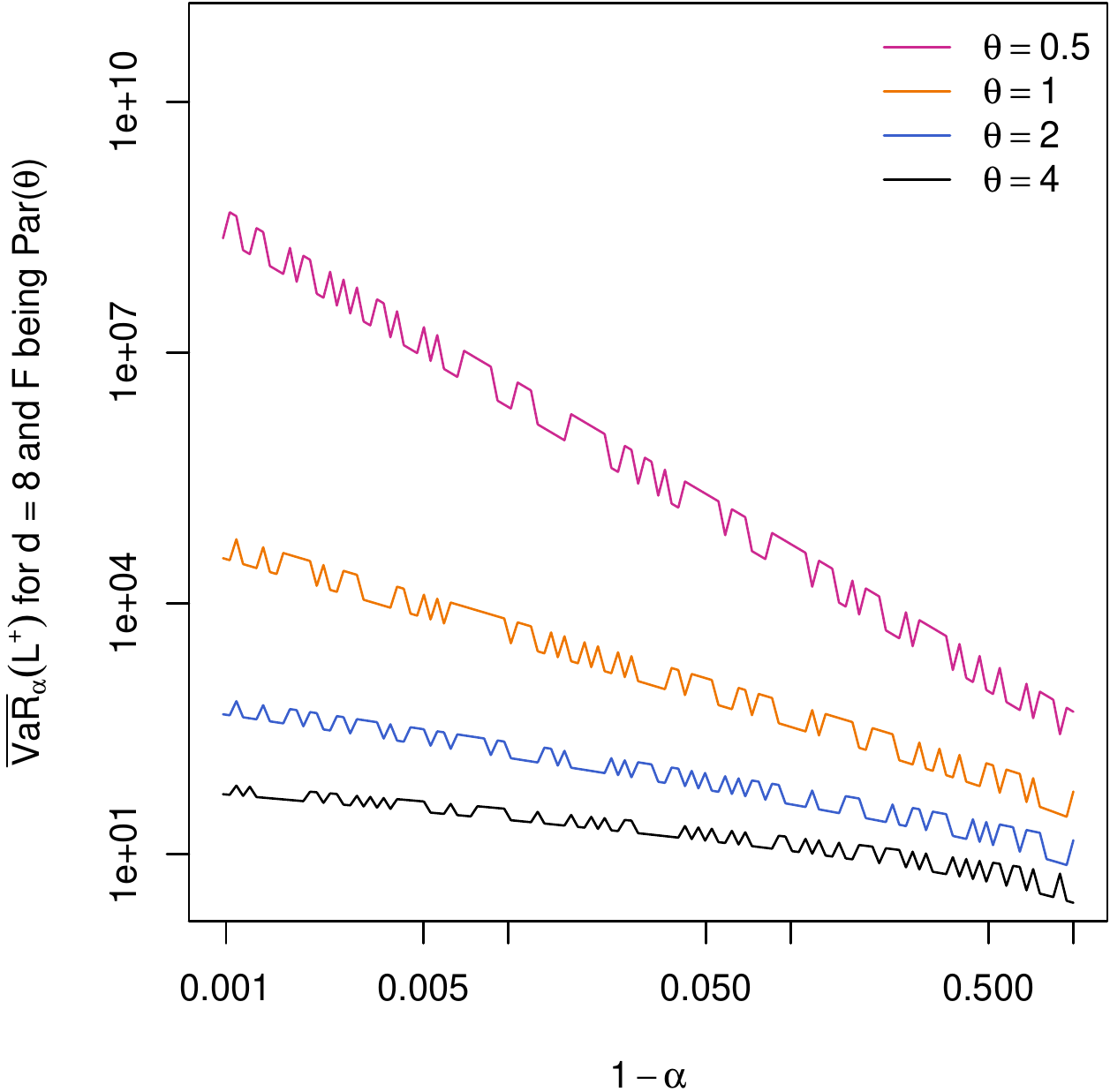}%
  \hfill
  \includegraphics[width=0.48\textwidth]{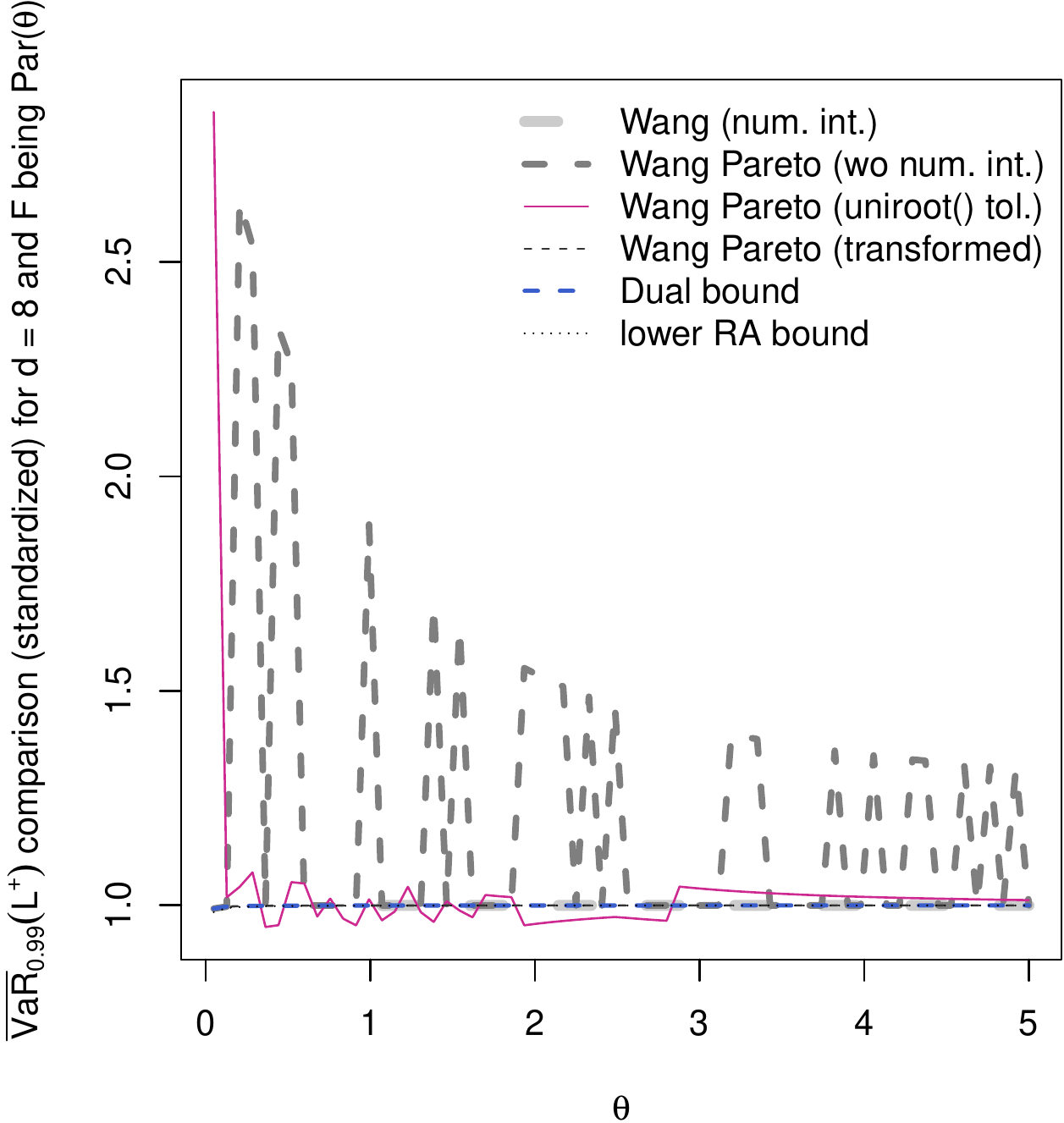}%
  \caption{Figures corresponding to the left-hand side of
    Figures~\ref{fig:worst:VaR:hom:Wang} ($\bVaR_\alpha(L^+)$ only) and
    \ref{fig:hom:comparison} (standardized with respect to the upper bound
    obtained from the RA) but for $h((1-\alpha)/d)$ adjusted to
    \texttt{.Machine\$double.xmin}.}
  \label{fig:worst:VaR:hom:Wang:wrong:right:adjustment}
\end{figure}

After carefully considering all the numerical issues, we can now look at
$\underline{\VaR}_\alpha(L^+)$ and $\bVaR_\alpha(L^+)$ from a different
perspective. The right-hand side of
Figure~\ref{fig:worst:VaR:hom:Wang:Par:numerics} shows
$\underline{\VaR}_\alpha(L^+)$ and $\bVaR_\alpha(L^+)$ as functions in the
dimension $d$. The linearity of $\bVaR_\alpha(L^+)$ in the log-log scale suggests
that $\bVaR_\alpha(L^+)$ is actually a power function in $d$. To the best of our
knowledge, this is not known (nor theoretically justified) yet.
\begin{figure}[htbp]
  \centering
  \includegraphics[width=0.48\textwidth]{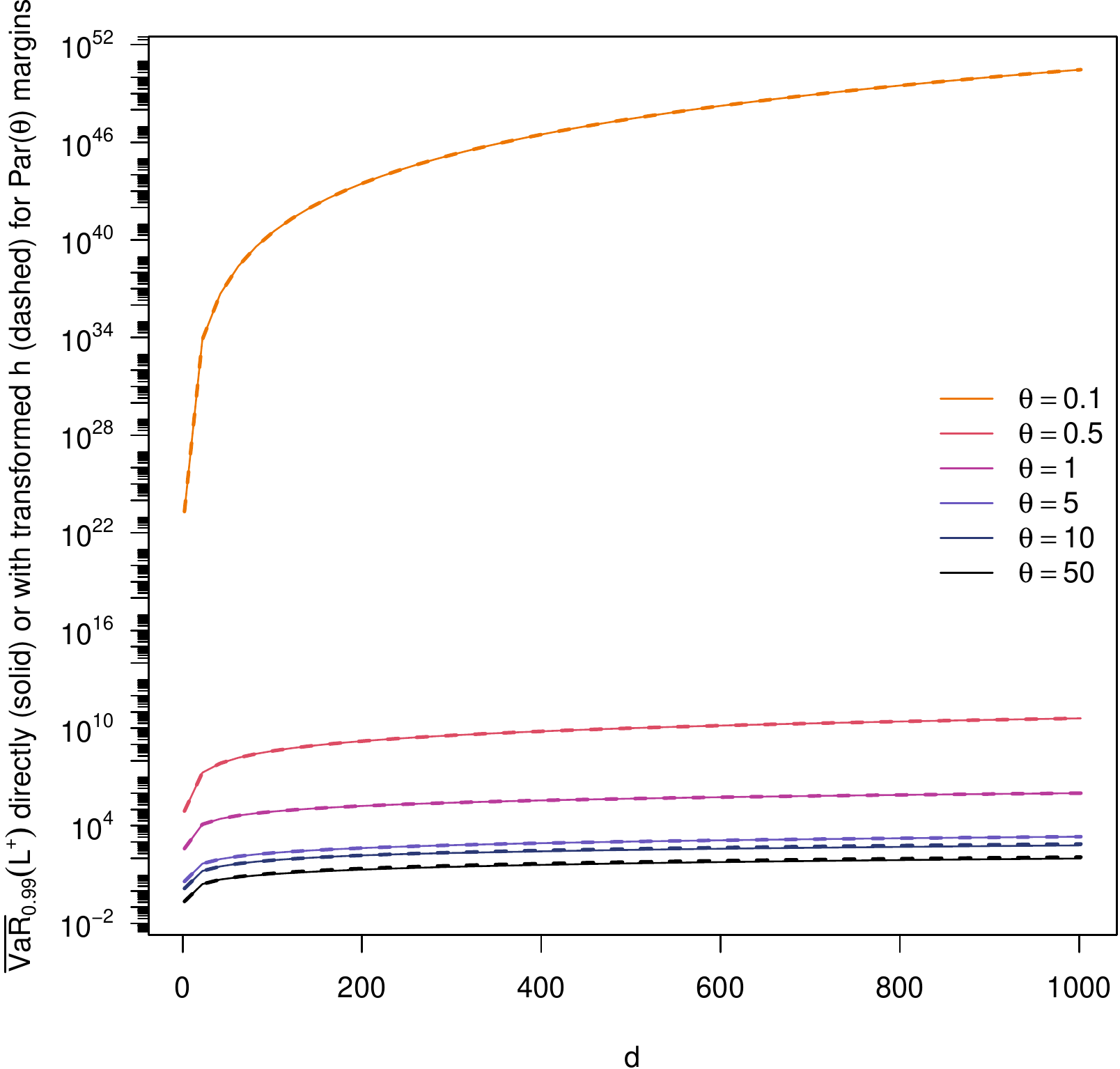}%
  \hfill
  \includegraphics[width=0.48\textwidth]{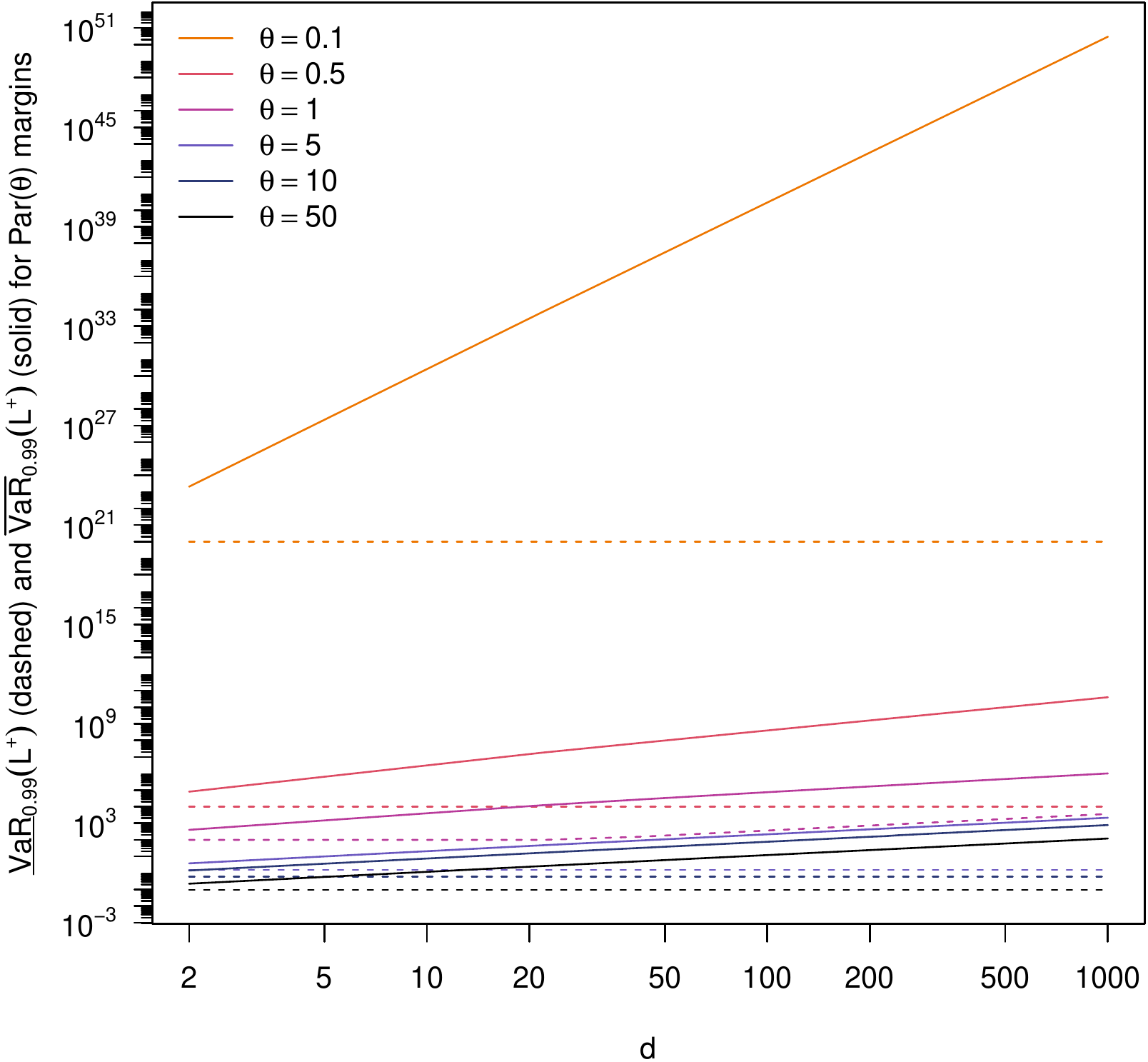}%
  \caption{A comparison of $\bVaR_\alpha(L^+)$ computed with \texttt{worst\_VaR\_hom(..., method="Wang.Par")}
    (solid line) and with the approach based on transforming the auxiliary
    function $h$ to a root-finding problem on $(1,\infty)$ (dashed line) as described in the
    proof of Proposition~\ref{prop:cl:cu:Par} (left-hand
    side). $\underline{\VaR}_\alpha(L^+)$ (dashed line) and $\bVaR_\alpha(L^+)$
    (solid line) as functions in $d$ on log-log scale (right-hand side).}
  \label{fig:worst:VaR:hom:Wang:Par:numerics}
\end{figure}

\section{The Rearrangement Algorithm}\label{sec:RA}
We now consider the RA for computing $\underline{\VaR}_\alpha(L^+)$ and $\bVaR_\alpha(L^+)$ in the inhomogeneous
case; as before, we mainly focus on $\bVaR_\alpha(L^+)$ here.

\subsection{About the algorithm}

One of the early works on finding bounds for $\VaR_\alpha(L^+)$ (including a
proof of their sharpness) can be found in \cite{makarov1982}, who provides bounds
on the distribution function of the sum of two random variables and bounds on
$\VaR_\alpha(L^+)$ thus follow. Later on, \cite{firporidder2010} prove these results using
copula theory, introducing dependence structures into the above framework and
extend Makarov's results to include an arbitrary increasing continuous
aggregation function (they do not prove the sharpness of the bounds,
though). \cite{williamsondowns1990} develop new methods for calculating convolutions
and dependency bounds for the distributions of functions of random variables; they
use lower and upper approximations to the desired distribution, containing
the representation error, and provide bounds on the errors.

\cite{denuitgenestmarceau1999} extend the above two-dimensional frameworks and
show how to compute bounds on the distribution function
of $L^+=L_1+\dots+L_d$ for the $d \ge 3$ case.
In a similar work, \cite{cossettedenuitmarceau2002} further develop results
about $(L_1,L_2)$ by assuming additional information about the correlation
structure of $(L_1,L_2)$. They further extend their results to the general
multivariate case and propose bounds for continuous and componentwise monotone
functions in $L_1,\dots, L_d$, assuming that the only available information on
$L_j$ is its distribution function $F_j$, $j\in\{1,\dots,d\}$. By relaxing some
of the continuity assumptions with respect to the aggregation function,
\cite{embrechts2003using} provide a generalization of these results
using copula theory.
In the latter paper, it is shown that without any prior
information on the dependence structure, only bounds for the
distribution function of the sum of the risks can be found and the problem of
the \emph{sharpness} of these bounds remained open when $d\ge 2$.

\cite{embrechtspuccetti2006b} provide better bounds on $F_{L^+}$ in the
\emph{homogeneous case} $F_1=\dots=F_d=F$, see Section~\ref{sec:dual:bd}, based
on the duality result of \cite{rueschendorf1982} for a continuous distribution
function $F$. The assumption of equal margins is rather restrictive,
especially for large $d$. To address this problem, \cite{embrechtspuccetti2006a}
extend the dual bound approach to general portfolios and describe a numerical
procedure to compute a lower bound for $F_{L^+}$ for an inhomogeneous portfolio of risks.  The
shortcoming of this method is 
that it requires the application of a global optimization algorithm for which there is no
guarantee of convergence to a global optimum
in a reasonable and predictable
amount of time; more importantly the quality of performance of many of these optimization
procedures is not well understood and in general the performance
of such algorithms deteriorates as $d$ increases. For these reasons the application of this method
for $d \geq 50$ becomes intractable in some cases.

The above problems can be studied in the context of \emph{completely mixable
  matrices}, i.e., matrices such that there exists a collection of permutations
acting on the columns which result in a constant row sum. This concept of
complete mixability is introduced and discussed in \cite{wang2011complete}. If a
matrix is not completely mixable, determining the smallest maximal and largest
minimal row sums is of interest to find $\underline{\VaR}_{\alpha}(L^+)$ and
$\bVaR_{\alpha}(L^+)$, respectively (or at least approximations of such). These \emph{minimax}
and \emph{maximin problems}, respectively, are in turn connected to discrete
approximations of the marginal quantile functions. \cite{haus2014} points out
that such approaches for computing $\underline{\VaR}_{\alpha}(L^+)$ and
$\bVaR_{\alpha}(L^+)$ are related to the multidimensional bottleneck assignment
problem and complete mixability is in general $\mathcal{NP}$-complete. As a
result, convergence to an optimal solution in polynomial time is not guaranteed.

If we discretize the tail of each of the $d$ risk factors using $N$ points, the
underlying space over which the above problems are analyzed becomes an
$N \times d$ matrix and enumeration of the total number of possible matrices
arising from permuting all but one column becomes intractable in applications as
there are $(N!)^{d-1}$ possible
matrices. 
One can easily observe that for a portfolio of only 10 positions and $N=20$, the
total number of such matrices is $(20!)^9\in\mathcal{O}(10^{165})$. Note that
the choices of both $N=20$ and $d=10$ are extremely conservative and for
illustrative purposes only; in practice $N$ can easily be as large as $1000$ to
$100,\!000$ and many portfolios consist of at least 20 to 40 instruments; larger
financial institutions can have portfolios with several thousand instruments.

As noted earlier, the complexity of the optimization procedures required for finding dual bounds,
given arbitrary marginal distributions along with a high run time were the main drawbacks
of using the dual bounds approach for obtaining a lower and upper bound
for $\VaR_\alpha(L^+)$. As a result, \cite{puccettirueschendorf2012} propose the
Rearrangement Algorithm (RA) to tackle these issues. The initial idea underlying
the RA and the numerical approximation introduced in
\cite{puccettirueschendorf2012} for calculating $\underline{\VaR}_{\alpha}(L^+)$
and $\bVaR_{\alpha}(L^+)$ is due to \cite{rueschendorf1983a} and
\cite{rueschendorf1983b}, respectively. The higher accuracy of the RA
(theoretically still an open question) along with its simple implementation
compared to the previous methods makes the RA an attractive alternative for
obtaining $\underline{\VaR}_{\alpha}(L^+)$ and $\bVaR_{\alpha}(L^+)$
when (only) the marginal loss distributions $F_1,\dots,F_d$ are known. In the
following subsections we look at how the RA works and analyze its performance using
various test cases.

\subsection{How the Rearrangement Algorithm works}

The RA can be applied to approximate the best Value-at-Risk
$\underline{\VaR}_{\alpha}(L^+)$ or the worst Value-at-Risk
$\bVaR_{\alpha}(L^+)$ for any set of marginals $F_j$, $j\in\{1,\dots,d\}$.  In
what follows our focus is mainly on $\bVaR_{\alpha}(L^+)$; our implementation
\texttt{RA()} in the \R\ package \texttt{qrmtools} also addresses
$\underline{\VaR}_{\alpha}(L^+)$. To understand the algorithm, note that two
columns $\bm{a},\bm{b}\in\IR^N$ are called \emph{oppositely ordered} if for all
$i,j\in\{1,\dots,N\}$ we have $(a_i-a_j)(b_i-b_j) \le 0$. Given a number $N$ of
discretization points of the marginal quantile functions $F_1^-,\dots,F_d^-$
above $\alpha$ (see Steps~\ref{RA:mat:low} and \ref{RA:mat:up} of Algorithm~\ref{algo:RA} below),
the RA constructs two $(N, d)$-matrices, denoted by $\underline{X}^{\alpha}$ and
$\overline{X}^{\alpha}$; the first matrix aims at
constructing an approximation of $\bVaR_{\alpha}(L^+)$ from below, the second
matrix is used to construct an approximation of $\bVaR_{\alpha}(L^+)$ from
above. Separately for each of these matrices, the RA iterates over its columns
and permutes each of them so that it is oppositely ordered to the sum of all
other columns. This is repeated until the minimal row sum
\begin{align*}
  s(X)=\min_{1 \le i \le N}\sum_{1\le j \le d} x_{ij}
\end{align*}
(for $X=(x_{ij})$ being one of the said $(N,d)$-matrices) changes by less than a given
\emph{(convergence) tolerance} $\eps\ge0$. The RA for $\bVaR_{\alpha}(L^+)$ thus
aims at solving the maximin problem. The intuition behind this is to
minimize the variance of the conditional distribution of
$L^+|L^+>F_{L^+}^-(\alpha)$ to concentrate more of the $1-\alpha$ mass of
$F_{L^+}$ in its tail. This pushes $\VaR_{\alpha}(L^+)$ further up.
As \cite{embrechtspuccettirueschendorf2013} state, one then typically ends up
with two matrices whose minimal row sums are close to each other and roughly
equal to $\bVaR_{\alpha}(L^+)$. Note that if one iteration over all columns of
one of the matrices does not lead to any change in that matrix, then each column
of the matrix is oppositely ordered to the sum of all others and thus there is
also no change in the minimal row sum (but the converse is not necessarily true,
see below).

The version of the RA given below contains slightly more information than in
\cite{embrechtspuccettirueschendorf2013}; e.g., how infinite quantiles are dealt
with. For more features of the actual implementation which is an improved
version of the one given below (see Section~\ref{sec:RA:our:impl}), see \texttt{RA()} and the underlying
workhorse \texttt{rearrange()}. This includes, e.g., a parameter
\texttt{max.ra} determining the maximal number of columns to rearrange or the
choice $\eps=$\texttt{(abstol=)NULL} to iterate until each column is oppositely
ordered to the sum of all others. The latter is typically (by far) not implied
by $\eps=0$, but does not bring better accuracy (see, e.g., the application
discussed in the vignette \texttt{VaR\_bounds}) and is typically very time-consuming (hence the
introduction of \texttt{max.ra}).
\begin{algorithm}[RA for computing $\bVaR_{\alpha}(L^+)$]\label{algo:RA}
  \begin{enumerate}
  \item Fix a confidence level $\alpha\in(0,1)$, marginal quantile functions
    $F_1^-,\dots,F_d^-$, an integer $N\in\IN$ and the desired (absolute)
    convergence tolerance $\eps\ge0$.
  \item Compute the lower bound:
    \begin{enumerate}
    \item\label{RA:mat:low} Define the matrix
      $\underline{X}^{\alpha}=(\underline{x}_{ij}^{\alpha})$ for
      $\underline{x}_{ij}^{\alpha} =
      F^-_j\bigl(\alpha+\frac{(1-\alpha)(i-1)}{N}\bigr)$, $i\in\{1,\dots,N\}$,
      $j\in\{1,\dots,d\}$.
    \item\label{RA:perm:low} Permute randomly the elements in each column of $\underline{X}^{\alpha}$.
    \item\label{RA:step:4:low} For $1\le j\le d$, permute the $j$-th column of the matrix
      $\underline{X}^{\alpha}$ so that it becomes oppositely ordered to the sum of
      all other columns. Call the resulting matrix $\underline{Y}^{\alpha}$.
    \item\label{RA:step:5:low} Repeat Step~\ref{RA:step:4:low} until $s(\underline{Y}^{\alpha})-s(\underline{X}^{\alpha})\le\eps$, then set $\underline{s}_N = s(\underline{Y}^{\alpha})$.
    \end{enumerate}
  \item Compute the upper bound:
    \begin{enumerate}
    \item\label{RA:mat:up} Define the matrix $\overline{X}^{\alpha}=(\overline{x}_{ij}^{\alpha})$
      for
      $\overline{x}_{ij}^{\alpha}=F^-_j\bigl(\alpha+\frac{(1-\alpha)i}{N}\bigr)$,
      $i\in\{1,\dots,N\}$, $j\in\{1,\dots,d\}$. If (for $i=N$ and) for any
      $j\in\{1,\dots,d\}$, $F^-_j(1)=\infty$, adjust it to
      $F^-_j\bigl(\alpha+\frac{(1-\alpha)(N-1/2)}{N}\bigr)$.
    \item\label{RA:perm:up} Permute randomly the elements in each column of $\overline{X}^{\alpha}$.
    \item\label{RA:step:4:up} For $1\le j\le d$, iteratively rearrange the $j$-th column of the matrix
      $\overline{X}^{\alpha}$ so that it becomes oppositely ordered to the sum of
      all other columns. Call the resulting matrix $\overline{Y}^{\alpha}$.
    \item\label{RA:step:5:up} Repeat Step~\ref{RA:step:4:up} until $s(\overline{Y}^{\alpha})-s(\overline{X}^{\alpha})\le\eps$, then set $\overline{s}_N = s(\overline{Y}^{\alpha})$.
    \end{enumerate}
  \item Return $(\underline{s}_N,\ \overline{s}_N)$.
  \end{enumerate}
\end{algorithm}

As mentioned before, the main feature of the RA is to iterate over all columns
and oppositely order each of them with respect to the sum of all others (see
Steps~\ref{RA:step:4:low} and \ref{RA:step:4:up}). This procedure aims at reducing the
variance of the row sums with each rearrangement. Note that it does not necessarily
reach an optimal solution of the maximin problem (see, e.g.,
\cite[Lemma~6]{haus2014} for a counter-example) and thus the convergence
$|\overline{s}_N-\underline{s}_N|\to 0$ is not guaranteed. In order to reduce
the possibility of this happening in practice, the randomization of the initial
input in Steps~\ref{RA:perm:low} and \ref{RA:perm:up} has been put in place;
however, see our study in Section 2.3 of the vignette \texttt{VaR\_bounds}
concerning the possible rearranged output matrices and the influence of the
underlying sorting algorithm on the outcome.

\subsection{Conceptual and numerical improvements}\label{sec:RA:our:impl}
Some words of warning are in order here. Besides the confidence level $\alpha$
and the marginal quantile functions $F_1^-,\dots,F_d^-$, RA relies on two
sources of input, namely $N\in\IN$ and $\eps\ge0$, for which
\cite{embrechtspuccettirueschendorf2013} do not provide guidance on reasonable
defaults. Concerning $N$, it obviously needs to be ``sufficiently large'', but a
practitioner is left alone with such a choice. Another issue is the use of
the \emph{absolute} tolerance $\eps$ in the algorithm. There are two problems. The first
problem is that it is more natural to use a relative instead of an absolute tolerance
in this context. Without (roughly) knowing the minimal row sum in
Steps~\ref{RA:step:5:low} and \ref{RA:step:5:up}, a pre-specified absolute tolerance
does not guarantee that the change in the minimal row sum from
$\underline{X}^{\alpha}$ to $\underline{Y}^{\alpha}$ is of the right order (and
such order depends at least on $d$ and the chosen quantile functions). If $\eps$
is chosen too large, the computed bounds $\underline{s}_N$ and $\overline{s}_N$
would carry too much uncertainty, whereas if it is too small, an
unnecessarily long run time results; the latter seems to be the case for
\cite[Table~3]{embrechtspuccettirueschendorf2013}, where the chosen $\eps=0.1$
is roughly 0.000004\% of the computed $\bVaR_{0.99}(L^+)$.

The second problem is that the absolute tolerance $\eps$ is only used for
checking \emph{individual} ``convergence'' of $\underline{s}_N$ and of
$\overline{s}_N$. It does not guarantee that $\underline{s}_N$ and
$\overline{s}_N$ are sufficiently close to obtain a reasonable approximation to
$\bVaR_{\alpha}(L^+)$. We are aware of the theoretical hurdles underlying the
algorithm which are still open questions at this point (e.g., the probability of
convergence of $\underline{s}_N$ and $\overline{s}_N$ to $\bVaR_{\alpha}(L^+)$
or that $\bVaR_{\alpha}(L^+)\le\overline{s}_N$ for sufficiently large $N$), but
from a computational point of view one should still check that $\underline{s}_N$
and $\overline{s}_N$ are close to each other. Also, the algorithm should return
convergence and other useful information, e.g., the \emph{relative rearrangement
  range} $|(\overline{s}_N-\underline{s}_N)/\overline{s}_N|$, the actual
individual absolute tolerances reached when computing $\underline{s}_N$
and $\overline{s}_N$, the number of column rearrangements used,
logical variables indicating whether the individual absolute tolerances have
been reached, the number of column rearrangements for $\underline{s}_N$
and $\overline{s}_N$, the row sums computed after each column rearrangement,
the constructed input matrices $\underline{X}^{\alpha}$, $\overline{X}^{\alpha}$
and the corresponding rearranged, final matrices $\underline{Y}^{\alpha}$,
$\overline{Y}^{\alpha}$; see \texttt{RA()} in the \R\ package \texttt{qrmtools} for
such information.

Another suboptimal design of the RA is to iterate over all $d$ columns before
checking the termination conditions; see Steps~\ref{RA:step:4:low} and
\ref{RA:step:4:up} of Algorithm~\ref{algo:RA}. Our underlying workhorse
\texttt{rearrange()} keeps track of the column rearrangements of the last $d$
considered columns and can thus terminate after rearranging any column (not only
the last one); see also Algorithm~\ref{algo:ARA} below. This saves run time (despite the ``tracking''
overhead). We advise the interested reader to have a look at the source code of
\texttt{rearrange()} for further numerical and run-time improvements (fast accessing of
columns via lists; avoiding having to compute the row sums over all but the
current column; an extended tracing feature), some of which are mentioned in the
vignette \texttt{VaR\_bounds}.

\subsection{Empirical performance under various scenarios}\label{sec:simu:RA}
In order to empirically investigate the performance of the RA, we consider two
studies, each of which addresses four cases; we thus consider eight scenarios.
As studies, we consider the following:
\begin{enumerate}[label=Study~\arabic*]
 \item\hspace{-1ex}:\label{study:1} $N\in\{2^7,2^8,\dots,2^{17}\}$ and $d=20$;
 \item\hspace{-1ex}:\label{study:2} $N=2^8=256$ and $d\in\{2^2,2^3,\dots,2^{10}\}$.
\end{enumerate}
These choices allow us to investigate the impact of the upper tail
discretization parameter $N$ (in \ref{study:1}) and the impact of the number of risk
factors $d$ (in \ref{study:2}) on the performance of the RA. As cases, we consider the
following different marginal tail behaviors based on the Pareto
distribution function $F_j(x)=1-(1+x)^{-\theta_j}$:
\begin{enumerate}[label=xxxxxxxxx]
\item[Case HH]\hspace{-2ex}: $\theta_1,\dots,\theta_d$ form an equidistant sequence from 0.6 to 0.4;
  this case represents a portfolio with all marginal loss distributions being
  heavy-tailed (slightly increasing in heavy-tailedness).
\item[Case LH]\hspace{-2ex}: $\theta_1,\dots,\theta_d$ form an equidistant sequence from
  1.5 to 0.5; this case represents a portfolio with different marginal tail
  behaviors ranging from comparably light-tailed to very heavy-tailed
  distributions.
\item[Case LL]\hspace{-2ex}: $\theta_1,\dots,\theta_d$ form an equidistant sequence from
  1.6 to 1.4; this case represents a portfolio with all marginal loss
  distributions being comparably light-tailed.
\item[Case LH$_1$]\hspace{-2ex}: $\theta_1,\dots,\theta_{d-1}$ are chosen as in
  Case~LL and
  $\theta_d=0.5$; this case represents a portfolio all marginal loss
  distributions being light-tailed except the last.
\end{enumerate}
To keep the studies tractable, we focus on the confidence level $\alpha=0.99$
and the absolute convergence tolerance $\eps=0$ in all scenarios. Furthermore, we consider
$B=100$ replicated simulation runs in order to provide empirical 95\% confidence
intervals for the estimated quantities; note that some of them are so tight that
they are barely visible in the figures presented below. The $B$ replications only differ due to
different permutations of the columns in Steps~\ref{RA:perm:low} and
\ref{RA:perm:up} of Algorithm~\ref{algo:RA}, everything else is deterministic;
this allows us to study the effect of these (initial) randomization steps on the
(convergence) results of the RA. Concerning the hardware used, all results were
produced on an AMD 3.2 GHz Phenom II X4 955 processor with 8 GB RAM.

\subsubsection*{Results of Study 1 ($N$ running, $d$ fixed)}
The simulation results for \ref{study:1} can be summarized as follows:
\begin{itemize}
\item As can be seen in Figure~\ref{fig:RA:study:1:VaR}, the
  means over all $B$ computed $\underline{s}_N$ and $\overline{s}_N$ converge as $N$
  increases.
\item Figure~\ref{fig:RA:study:1:runtime} indicates that as $N$ increases, so
  does the mean elapsed time (as to be expected). Overall, run time does not
  drastically depend on the case for our choices of Pareto margins, which is a good feature.
\item Figure~\ref{fig:RA:study:1:num:iter} shows that the
  maximum number of column rearrangements rarely exceeds $10d$ as $N$ increases; this
  will be used as a default maximum number of column rearrangements
  required in the ARA.
\item Finally, Figure~\ref{fig:RA:study:1:num:opp:ordered} indicates that
  the rate at which the number of oppositely ordered columns (based on the final
  rearranged $\underline{Y}^{\alpha}$ and $\overline{Y}^{\alpha}$) decreases depends
  on the characteristics of the marginal distributions involved, i.e., the input
  matrix $X$. The number of oppositely ordered columns seems particularly small
  (for large $N$) in Case~LL, where essentially only the last column is
  oppositely ordered to the sum of all others.
\end{itemize}

\begin{figure}[htbp]
  \centering
  \includegraphics[width=\textwidth]{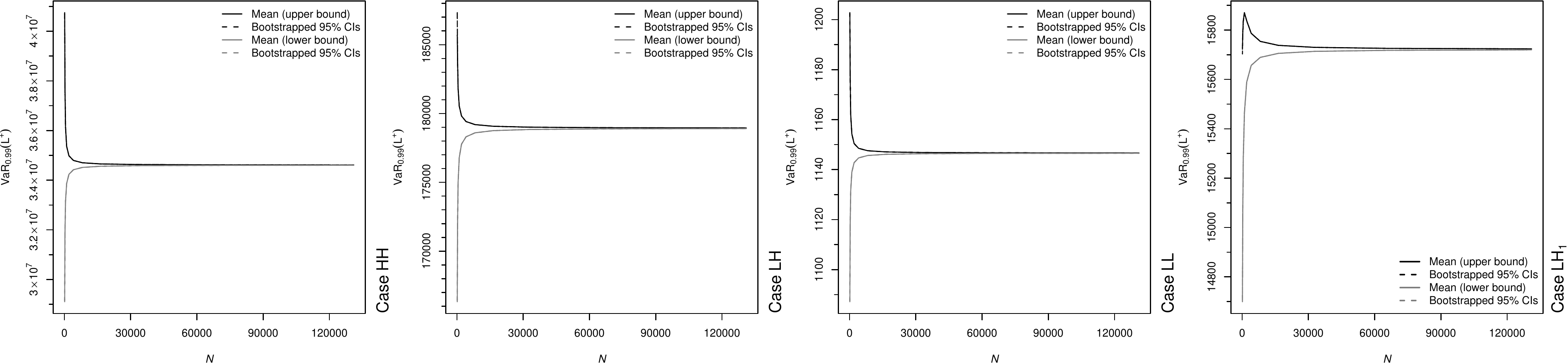}%
  \caption{Study 1: $\bVaR_{0.99}$ bounds $\underline{s}_N$ and $\overline{s}_N$
    for the Cases~HH, LH, LL and LH$_1$ (from left to right).}
  \label{fig:RA:study:1:VaR}
\end{figure}
\begin{figure}[htbp]
  \centering
  \includegraphics[width=\textwidth]{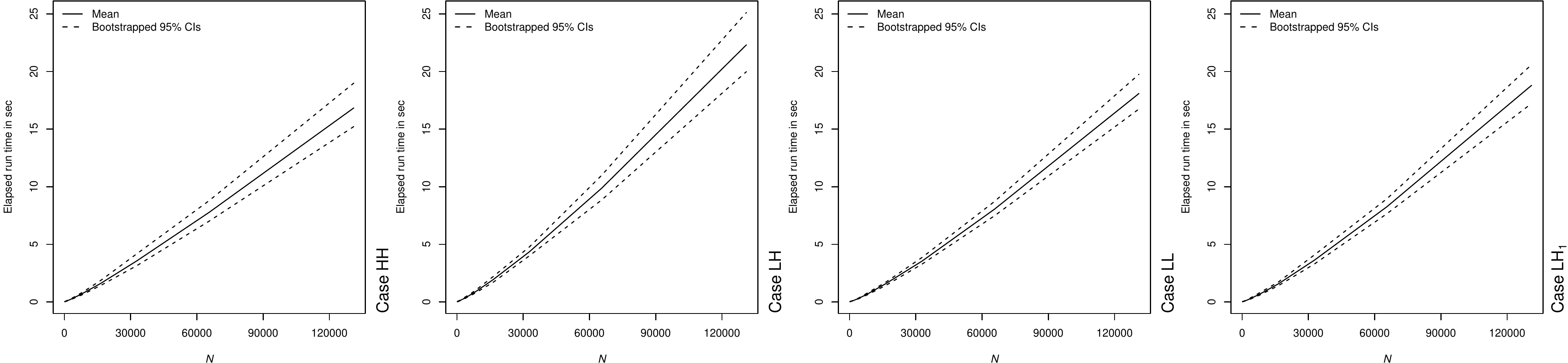}%
  \caption{Study 1: Run times (in s) for the Cases~HH, LH, LL and LH$_1$ (from left to right).}
  \label{fig:RA:study:1:runtime}
\end{figure}
\begin{figure}[htbp]
  \centering
  \includegraphics[width=\textwidth]{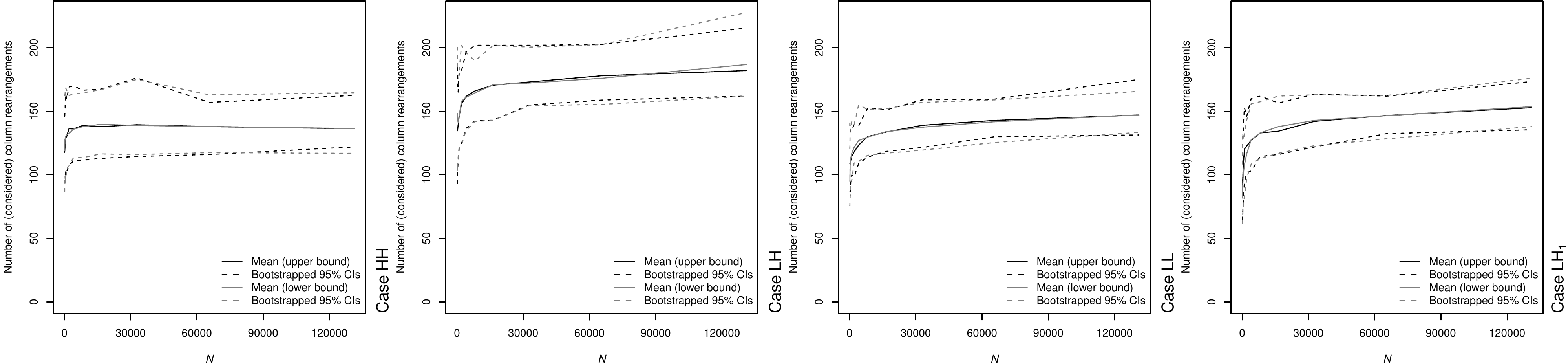}%
  \caption{Study 1: Number of rearranged columns for the Cases~HH, LH, LL
    and LH$_1$ (from left to right).}
  \label{fig:RA:study:1:num:iter}
\end{figure}
\begin{figure}[htbp]
  \centering
  \includegraphics[width=\textwidth]{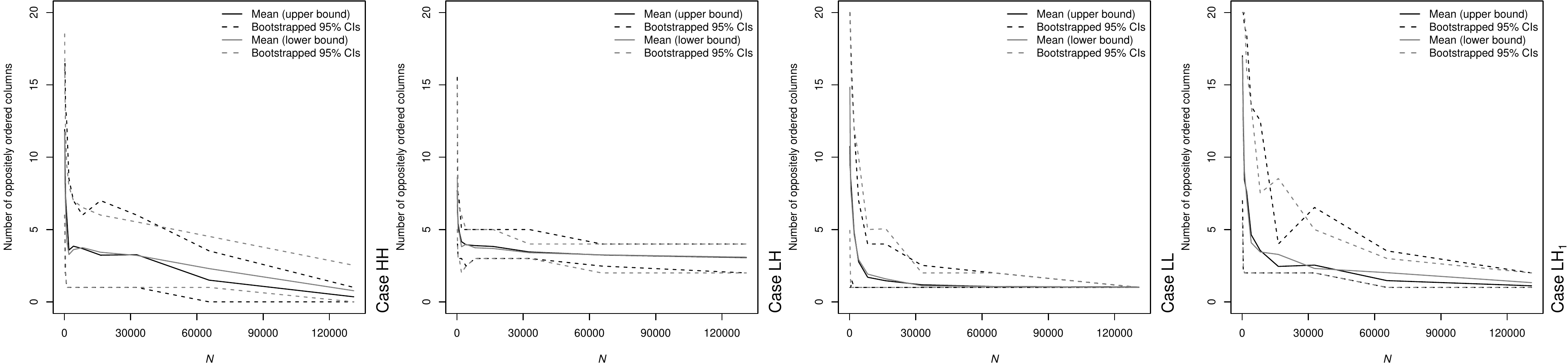}%
  \caption{Study 1: Number of oppositely ordered columns (of
    $\underline{Y}^{\alpha}$ and $\overline{Y}^{\alpha}$) for the Cases~HH, LH,
    LL and LH$_1$ (from left to right).}
  \label{fig:RA:study:1:num:opp:ordered}
\end{figure}

\subsubsection*{Results of Study 2 ($N$ fixed, $d$ running)}
Figures~\ref{fig:RA:study:2:VaR}--\ref{fig:RA:study:2:num:opp:ordered} show the
performance of the RA in \ref{study:2}; here we are interested in
analyzing the impact of the number of risk factors $d$ on portfolios which
exhibit different marginal tail behaviors. The simulation results can be
summarized as follows:
 \begin{itemize}
 \item Figure~\ref{fig:RA:study:2:VaR} shows that the means over all computed
   $\underline{s}_N$ and $\overline{s}_N$ diverge from one another; especially
   when more marginal distributions are heavy-tailed. This is due to the fact
   that we have kept $N$ the same for all cases in \ref{study:2}.
 \item Similar to what we have seen in \ref{study:1},
   Figure~\ref{fig:RA:study:2:runtime} indicates that the mean run time of the RA increases as the
   number of risk factors increases; Case~LL in \ref{study:2} has the
   least run time on average as $\theta_1,\dots,\theta_d$ form an equidistant
   sequence from 1.6 to 1.4 which results in smaller elements of the input
   matrix $X$ with a smaller range of entries compared to the other cases.
 \item As in \ref{study:1}, the mean number of rearranged columns is typically
   below $10d$; see
   Figure~\ref{fig:RA:study:2:num:iter}.
 \item The number of oppositely ordered columns (of $\underline{Y}^{\alpha}$ and $\overline{Y}^{\alpha}$), see
   Figure~\ref{fig:RA:study:2:num:opp:ordered}, increases as $d$
   increases; note that this finding does not contradict what we have seen in
   Study~1 as we have kept $N$ the same here.
 \end{itemize}

\begin{figure}[htbp]
  \centering
  \includegraphics[width=\textwidth]{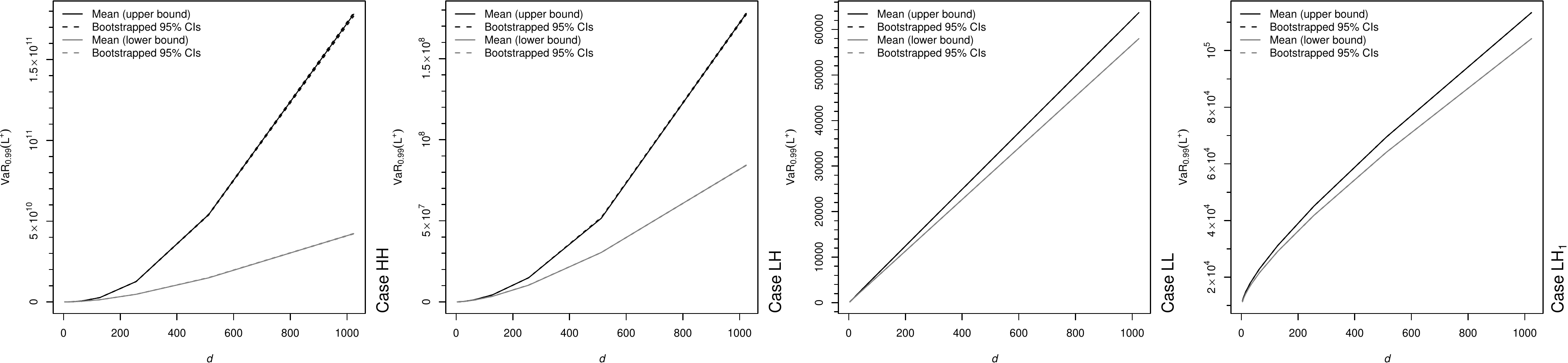}%
  \caption{Study 2: $\bVaR_{0.99}$ bounds
    $\underline{s}_N$ and $\overline{s}_N$ for the Cases~HH, LH, LL and LH$_1$
    (from left to right).}
  \label{fig:RA:study:2:VaR}
\end{figure}
\begin{figure}[htbp]
  \centering
  \includegraphics[width=\textwidth]{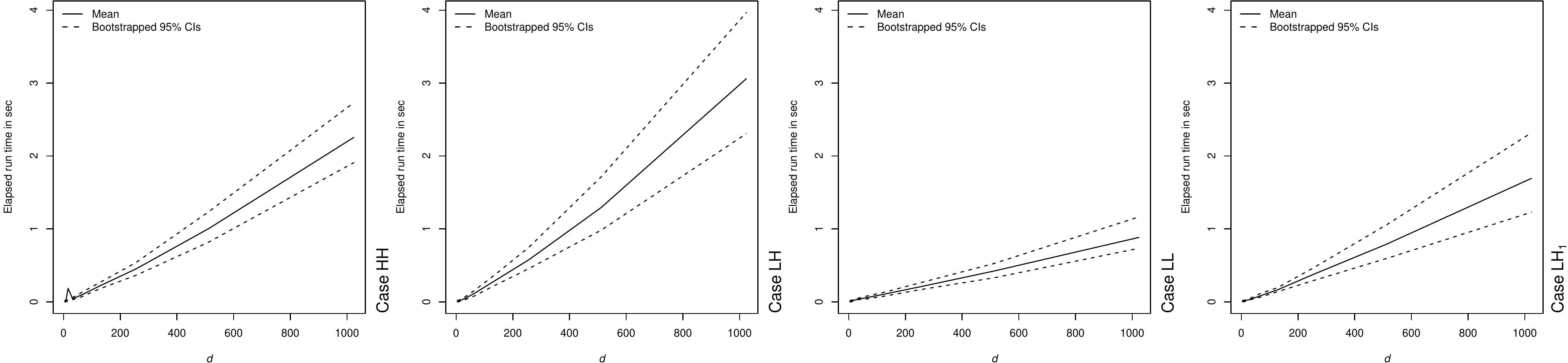}%
  \caption{Study 2: Run times (in s) for the Cases~HH, LH, LL and LH$_1$
    (from left to right).}
  \label{fig:RA:study:2:runtime}
\end{figure}
\begin{figure}[htbp]
  \centering
  \includegraphics[width=\textwidth]{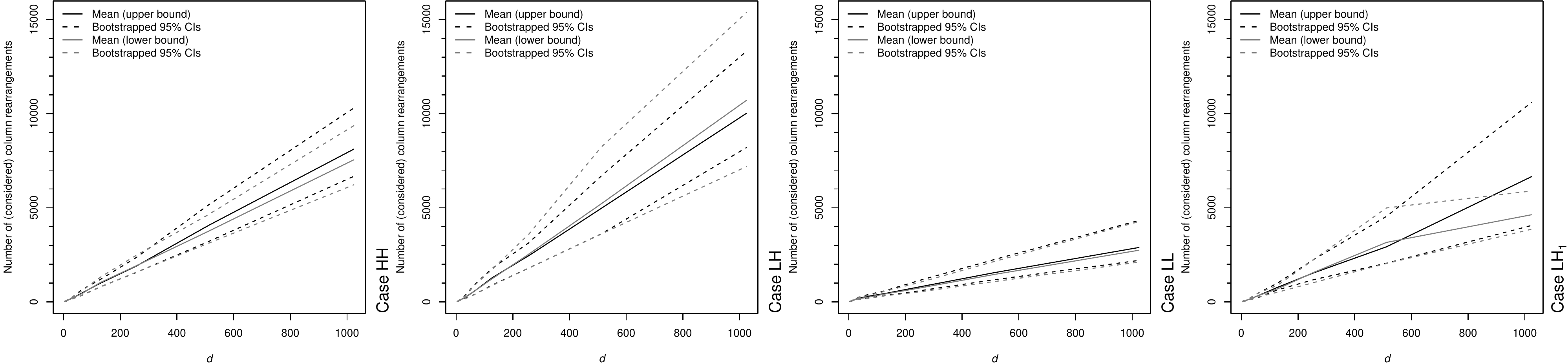}%
  \caption{Study 2: Number of rearranged columns for the Cases~HH, LH, LL and LH$_1$
    (from left to right).}
  \label{fig:RA:study:2:num:iter}
\end{figure}
\begin{figure}[htbp]
  \centering
  \includegraphics[width=\textwidth]{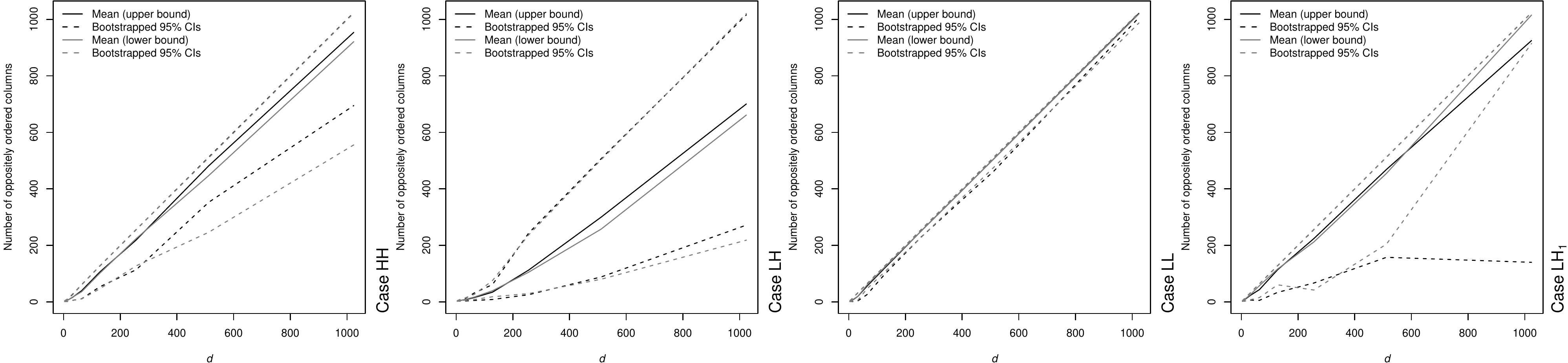}%
  \caption{Study 2: Number of oppositely ordered columns (of
    $\underline{Y}^{\alpha}$ and $\overline{Y}^{\alpha}$) for the Cases~HH, LH,
    LL and LH$_1$ (from left to right).}
  \label{fig:RA:study:2:num:opp:ordered}
\end{figure}

\section{The Adaptive Rearrangement Algorithm}\label{sec:ARA}
\subsection{How the Adaptive Rearrangement Algorithm works}
In this section we present an adaptive version of the RA, termed \emph{Adaptive
  Rearrangement Algorithm (ARA)}. This algorithm for computing the bounds
$\underline{s}_N$ and $\overline{s}_N$ for $\underline{\VaR}_{\alpha}(L^+)$ or
$\bVaR_{\alpha}(L^+)$ (as before, we focus on the latter) provides an
algorithmically improved version of the RA, has more meaningful tuning
parameters and adaptively chooses the number of discretization points $N$. The
ARA is implemented in the \R\ package \texttt{qrmtools}, see the function
\texttt{ARA()}. Similar to our \texttt{RA()} implementation, \texttt{ARA()} also
relies on the workhorse \texttt{rearrange()} and returns much more information
(and can also compute $\underline{\VaR}_{\alpha}(L^+)$),
but the essential part of the algorithm is given as follows.

\begin{algorithm}[ARA for computing $\bVaR_{\alpha}(L^+)$]\label{algo:ARA}
  \begin{enumerate}
  \item Fix a confidence level $\alpha\in(0,1)$, marginal quantile functions
    $F_1^-,\dots,F_d^-$, an integer vector $\bm{K}\in\IN^l$, $l\in\IN$,
    (containing the numbers of discretization points which are adaptively used),
    a bivariate vector of relative convergence tolerances $\bm{\eps}=(\eps_1,\eps_2)$
    (containing the individual relative tolerance $\eps_1$ and the joint relative
    tolerance $\eps_2$; see below) and the maximal number of iterations used for
    each $k\in\bm{K}$.
  \item For $N=2^k$, $k\in\bm{K}$, do:
    \begin{enumerate}
    \item Compute the lower bound:
      \begin{enumerate}
      \item Define the matrix
        $\underline{X}^{\alpha}=(\underline{x}_{ij}^{\alpha})$ for
        $\underline{x}_{ij}^{\alpha} =
        F^-_j\bigl(\alpha+\frac{(1-\alpha)(i-1)}{N}\bigr)$, $i\in\{1,\dots,N\}$,
        $j\in\{1,\dots,d\}$.
      \item Permute randomly the elements in each column of $\underline{X}^{\alpha}$.
      \item\label{ARA:step:4:low} Iteratively rearrange the $j$th column ($j\in\{1,2,\dots,d,1,2,\dots,d,\dots\}$)
        of $\underline{X}^{\alpha}$ oppositely to the sum of all other columns
        until the maximal number of column rearrangements has been reached or until
        \begin{align}
          \Bigl|\frac{s(\underline{Y}^{\alpha})-s(\underline{X}^{\alpha})}{s(\underline{X}^{\alpha})}\Bigr|\le\eps_1,\label{ARA:conv:cond:low}
        \end{align}
        where $\underline{Y}^{\alpha}$ denotes the matrix of quantiles after the
        $j$th column has been rearranged and $\underline{X}^{\alpha}$ denotes
        the same matrix $d$ steps earlier, i.e., after the $j$th column has been
        rearranged the last time.
      \item Set $\underline{s}_N = s(\underline{Y}^{\alpha})$.
      \end{enumerate}
    \item Compute the upper bound:
      \begin{enumerate}
      \item Define the matrix $\overline{X}^{\alpha}=(\overline{x}_{ij}^{\alpha})$
        for $\overline{x}_{ij}^{\alpha}=F^-_j\bigl(\alpha+\frac{(1-\alpha)i}{N}\bigr)$,
        $i\in\{1,\dots,N\}$, $j\in\{1,\dots,d\}$. If (for $i=N$ and) for any
        $j\in\{1,\dots,d\}$, $F^-_j(1)=\infty$, adjust it to
        $F^-_j\bigl(\alpha+\frac{(1-\alpha)(N-1/2)}{N}\bigr)$.
      \item Permute randomly the elements in each column of $\overline{X}^{\alpha}$.
      \item\label{ARA:step:4:up} Iteratively rearrange the $j$th column ($j\in\{1,2,\dots,d,1,2,\dots,d,\dots\}$)
        of $\overline{X}^{\alpha}$ oppositely to the sum of all other columns
        until the maximal number of column rearrangements has been reached or until
        \begin{align}
          \Bigl|\frac{s(\overline{Y}^{\alpha})-s(\overline{X}^{\alpha})}{s(\overline{X}^{\alpha})}\Bigr|\le\eps_1,\label{ARA:conv:cond:up}
        \end{align}
        where $\overline{Y}^{\alpha}$ denotes the matrix of quantiles after the
        $j$th column has been rearranged and $\overline{X}^{\alpha}$ denotes
        the same matrix $d$ steps earlier, i.e., after the $j$th column has been
        rearranged the last time.
      \item Set $\overline{s}_N = s(\overline{Y}^{\alpha})$.
      \end{enumerate}
    \item Determine convergence based on both the individual and the joint
      relative convergence tolerances:
      \begin{align*}
        \text{If \eqref{ARA:conv:cond:low} and \eqref{ARA:conv:cond:up} hold,
          and if}\
        \Bigl|\frac{\overline{s}_N-\underline{s}_N}{\overline{s}_N}\Bigr|\le\eps_2,\
        \text{break.}
      \end{align*}
    \end{enumerate}
  \item Return $(\underline{s}_N,\ \overline{s}_N)$.
  \end{enumerate}
\end{algorithm}

Concerning the choices of tuning parameters in Algorithm~\ref{algo:ARA}, note
that if $\bm{K}=(k=\log_2N)$, so if we have a single number of discretization
points, then the ARA reduces to the RA but uses the more meaningful relative
instead of absolute tolerances and not only checks what we termed
\emph{individual (convergence) tolerance}, i.e., the tolerance $\eps_1$ for
checking ``convergence'' of $\underline{s}_N$ and of $\overline{s}_N$
individually, but also the \emph{joint (convergence tolerance)}, i.e., the
relative tolerance $\eps_2$ between $\underline{s}_N$ and $\overline{s}_N$;
furthermore, termination is checked after each column rearrangement (which is
also done by our implementation \texttt{RA()} but not the version given in
\cite{embrechtspuccettirueschendorf2013}). As our simulation studies in
Section~\ref{sec:simu:RA} suggest, useful (conservative) defaults for $\bm{K}$
and the maximal number of column rearrangements are $\bm{K}=(8,9,\dots,19)$ and $10d$,
respectively. Given the high model uncertainty and the (often) rather large
values of $\bVaR_{\alpha}(L^+)$ (especially in heavy-tailed test cases), a
conservative choice for the relative tolerance $\bm{\eps}$ may be
$\bm{\eps}=(0,\ 0.01)$; obviously, all
these values can be freely chosen in the actual implementation of
\texttt{ARA()}.

\subsection{Empirical performance under various scenarios}
In this section, we empirically investigate the performance of the ARA. To this
end, we consider $d\in\{20,100\}$ risk factors, paired with the Cases~HH, LH,
LL, LH$_1$ as in Section~\ref{sec:simu:RA}. The considered relative joint
tolerances are 0.5\%, 1\% and 2\% and we investigate the results for the
individual relative tolerances 0\% and 0.1\%. Therefore the performance of ARA
is investigated in 48 different scenarios. As before, the results shown are
based on $B=100$ simulations and we investigate the $\bVaR_{0.99}(L^+)$ bounds
$\underline{s}_N$ and $\overline{s}_N$, the \emph{$N$ used} on the final
column rearrangement of ARA (i.e., the $N=2^k$, $k\in\bm{K}$ for which the algorithm terminates),
the run time (in s), the number of column rearrangements (measured for the $N$
used on termination of the algorithm) and the number of oppositely ordered
columns after termination of the algorithm.

\subsubsection*{Results}
We first consider the results for the individual relative tolerance chosen
as $\eps_1=0$.

Our findings (see Figures~\ref{fig:ARA:VaR:0}--\ref{fig:ARA:noi:0} and computed
numbers; the latter are omitted) indicate that:
\begin{itemize}
\item Although for both $d=20$ and $d=100$ the length of the confidence
  intervals for $\bVaR_{0.99}(L^+)$ can be checked to be increasing as the joint
  relative tolerance $\eps_2$ gets larger, the mean and lower and upper
  confidence bounds remain fairly close to each other. More importantly for a
  fixed individual relative tolerance, as $\eps_2$ increases, we do not observe
  a drastic shift in both lower and upper bounds for the mean across different
  examples.
\item An important observation about the $N$ in the ARA is
  that in virtually all examples, the 95\% confidence interval remains the same;
  this fact can be leveraged in practice for portfolios which exhibit the same
  marginal tail behavior to reduce the run time of the ARA.
\item Across all of the 24 scenarios, doubling the joint relative tolerance
  reduces the run time (measured in s) more than 50\%.
\item The number of column rearrangements for the $N$ used remains below $10d$.
\item Finally, as Figures~\ref{fig:ARA:VaR:0}--\ref{fig:ARA:noi:0} reveal,
  randomizing the input matrix $X$ has a minimal impact on various outputs of
  the ARA. However, this randomization has an interesting effect on the run time
  in that it seems to avoid the worst case in which sorting a lot of numbers is
  required when oppositely ordering the columns causing the algorithm to take quite a bit
  longer. This behavior (and dependence of run time on the order of the columns
  as well; typically, convergence was faster with heavier-tailed columns last)
  was clearly visible during our testing phase and leads to the following open
  research question: How can the rows and columns of the initial matrices be
  sorted such that the run time of the ARA is minimal?
\end{itemize}

The effect of a different choice of $\eps_1$ can be seen from
Figures~\ref{fig:ARA:VaR:0.001}--\ref{fig:ARA:noi:0.001}. Overall, they are very
similar, note however Figure~\ref{fig:worst:VaR:bounds:application} (based on
the 20 constituents of the SMI from 2011-09-12 to 2012-03-28) concerning a too
large choice of $\eps_1$. As we can see, both bounds can be fairly close
(according to $\eps_2$) but the individual ``convergence'' did not take place
yet; hence our default $\eps_1=0$ of \texttt{RA()} and \texttt{ARA()}.

\begin{figure}[htbp]
  \centering
  \includegraphics[width=\textwidth]{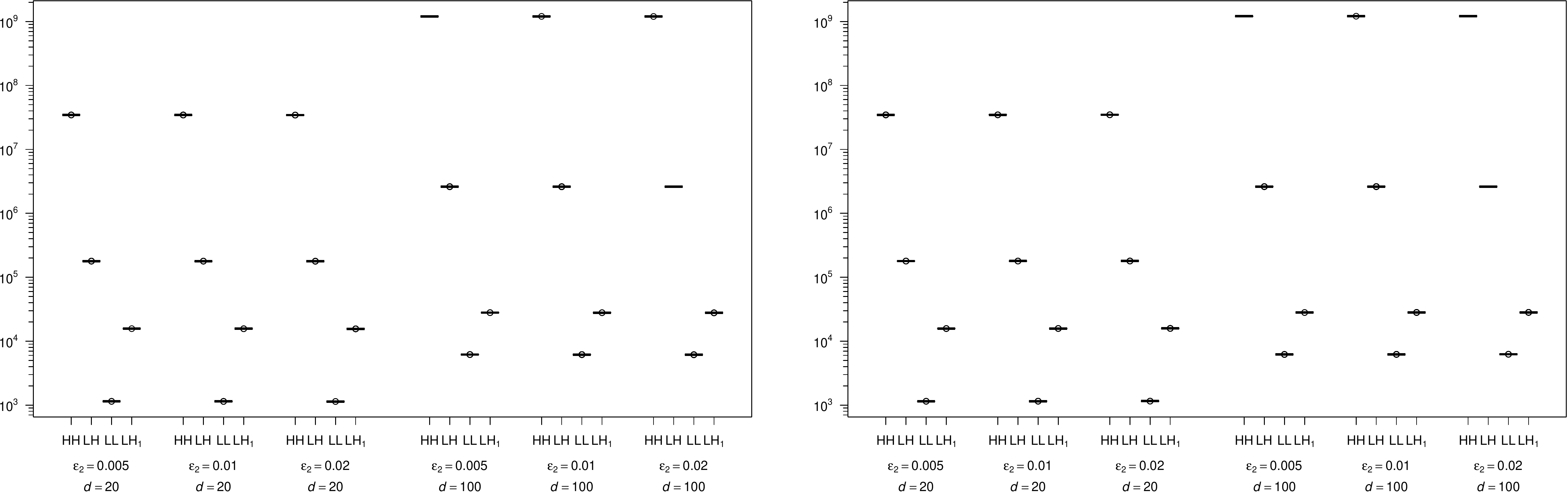}%
  \caption{Boxplots of the lower and upper $\bVaR_{0.99}(L^+)$ bounds
    $\underline{s}_N$ (left-hand side) and $\overline{s}_N$ (right-hand side)
    computed with the ARA for \texttt{itol=0} based on $B=100$ replications.}
  \label{fig:ARA:VaR:0}
\end{figure}

\begin{figure}[htbp]
  \centering
  \includegraphics[width=\textwidth]{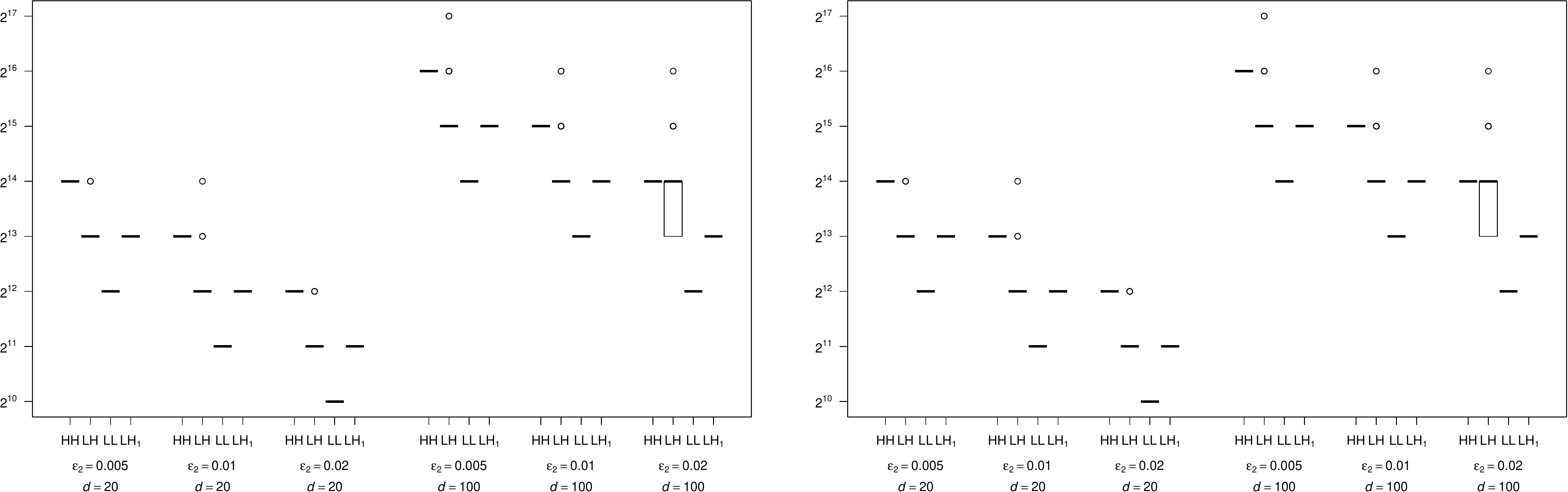}%
  \caption{Boxplots of the actual $N=2^k$ used for computing the lower and upper
    $\bVaR_{0.99}(L^+)$ bounds $\underline{s}_N$ (left-hand side) and
    $\overline{s}_N$ (right-hand side) with the ARA for \texttt{itol=0} based on $B=100$
    replications.}
  \label{fig:ARA:N:used:0}
\end{figure}

\begin{figure}[htbp]
  \centering
  \includegraphics[width=0.88\textwidth]{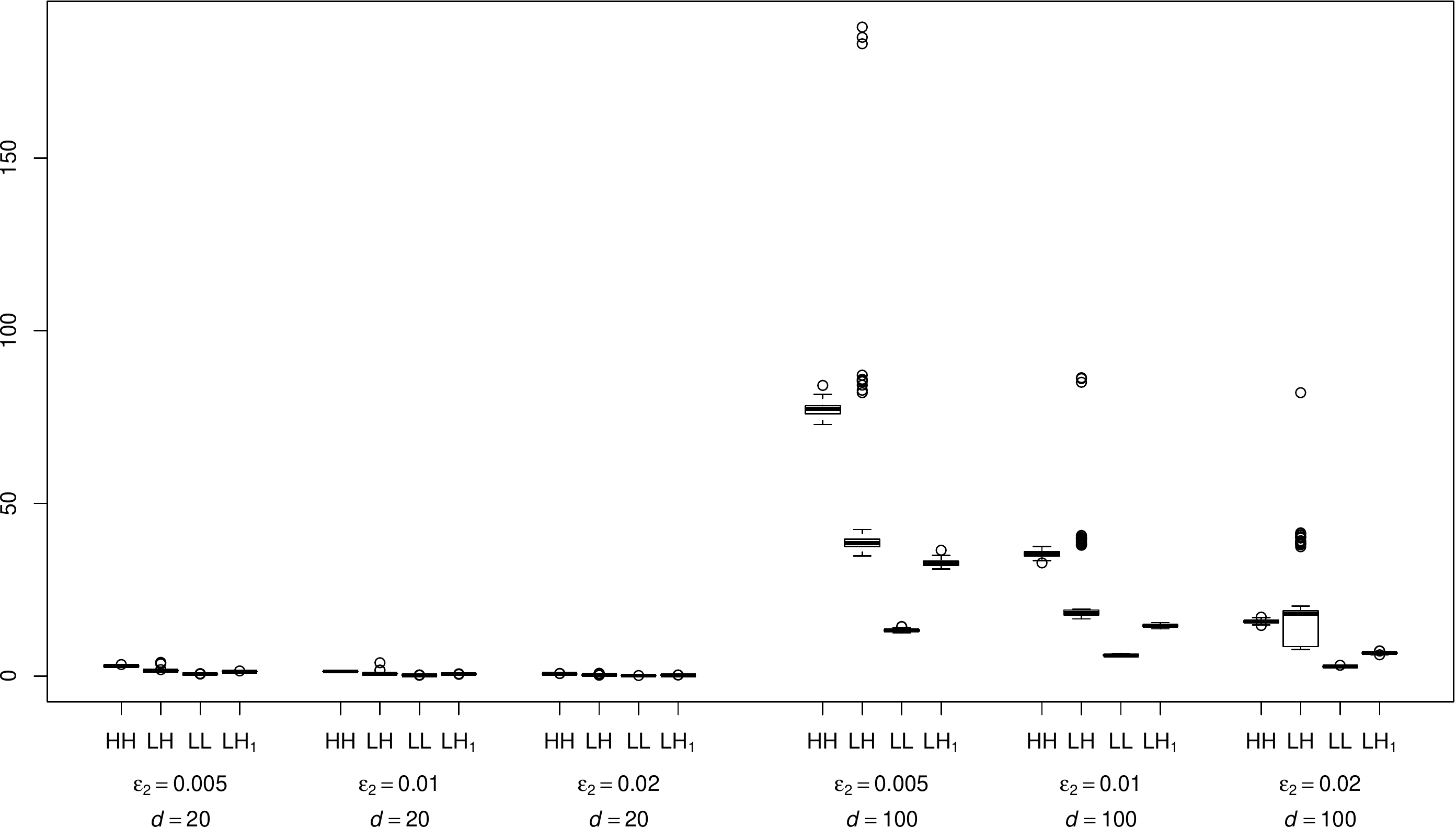}%
  \caption{Boxplots of the run time (in s) for computing the lower and upper
    $\bVaR_{0.99}(L^+)$ bounds $\underline{s}_N$ (left-hand side) and
    $\overline{s}_N$ (right-hand side) with the ARA for \texttt{itol=0} based on $B=100$
    replications.}
  \label{fig:ARA:runtime:0}
\end{figure}

\begin{figure}[htbp]
  \centering
  \includegraphics[width=\textwidth]{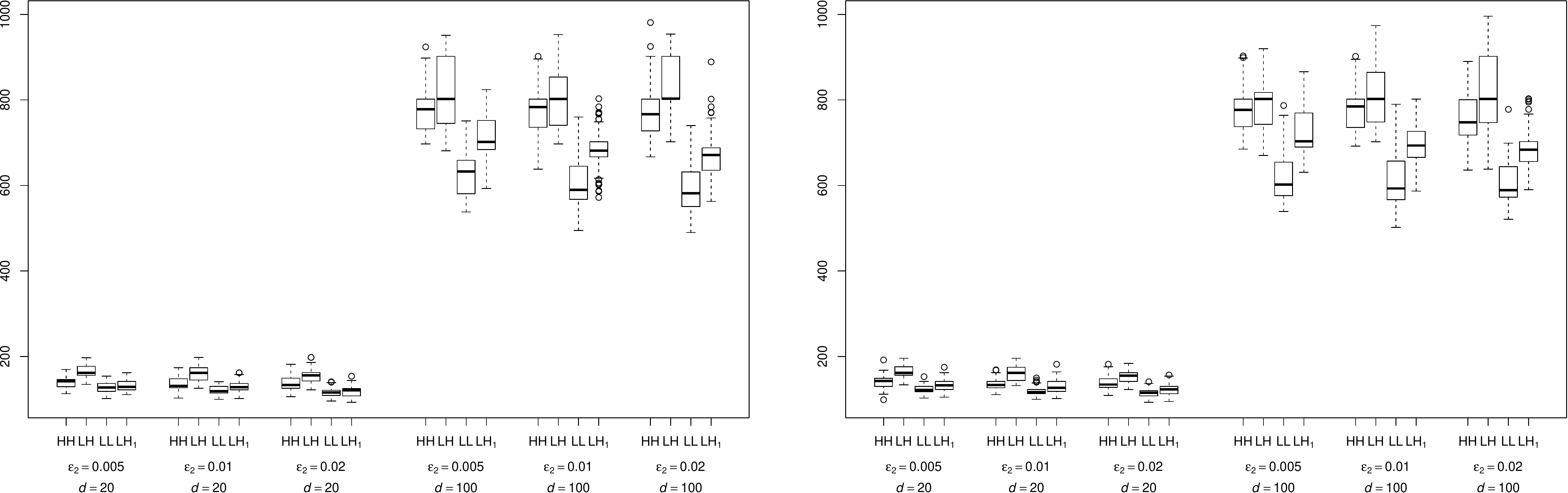}%
  \caption{Boxplots of the number of column rearrangements (measured for
    the $N$ used) for computing the lower and upper $\bVaR_{0.99}(L^+)$ bounds
    $\underline{s}_N$ (left-hand side) and $\overline{s}_N$ (right-hand side)
    with the ARA for \texttt{itol=0} based on $B=100$ replications.}
  \label{fig:ARA:noi:0}
\end{figure}

\begin{figure}[htbp]
  \centering
  \includegraphics[width=\textwidth]{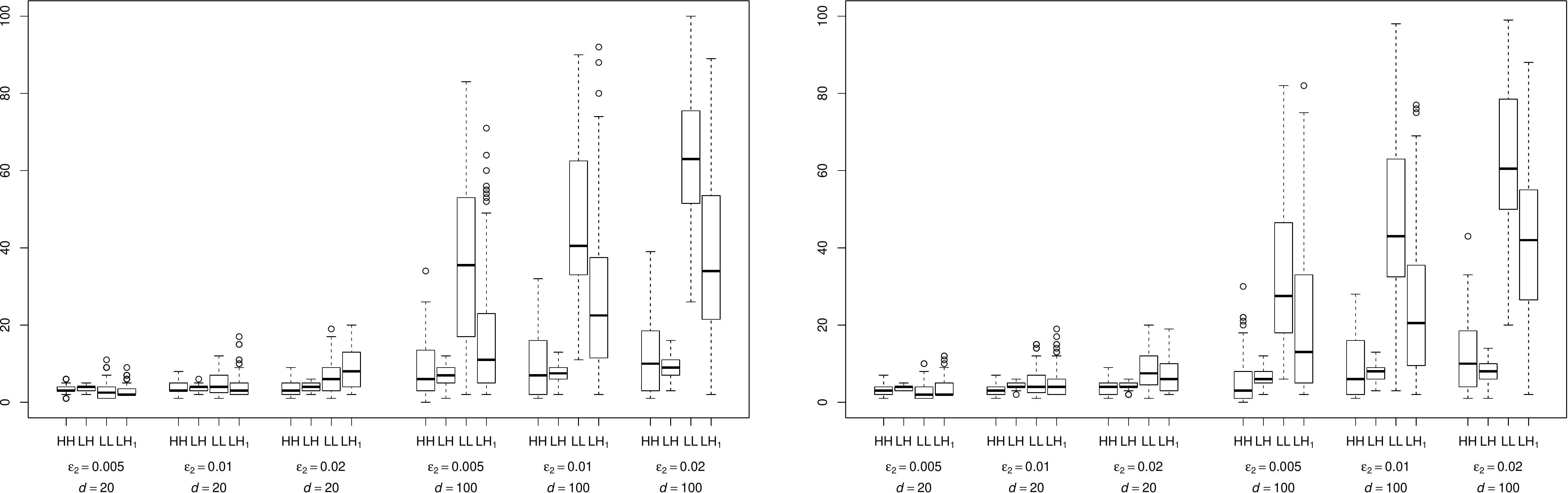}%
  \caption{Boxplots of the number of oppositely ordered columns for computing
    the lower and upper $\bVaR_{0.99}(L^+)$ bounds
    $\underline{s}_N$ (left-hand side) and $\overline{s}_N$ (right-hand side)
    with the ARA for \texttt{itol=0} based on $B=100$ replications.}
  \label{fig:ARA:noo:0}
\end{figure}

\begin{figure}[htbp]
  \centering
  \includegraphics[width=\textwidth]{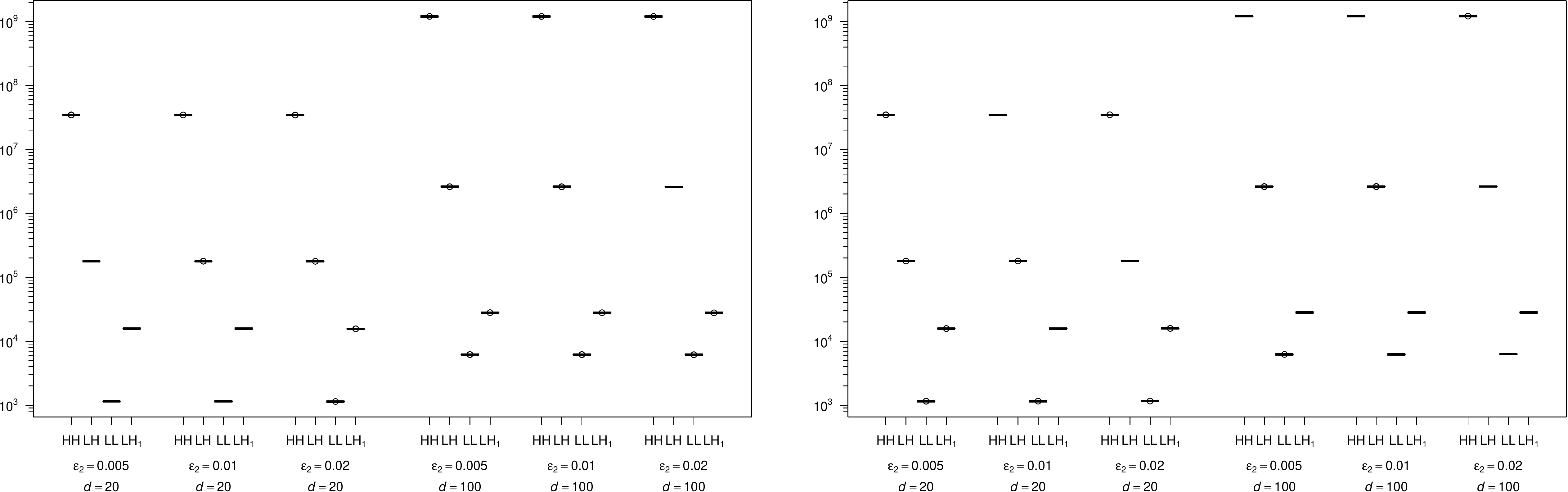}%
  \caption{Boxplots of the lower and upper $\bVaR_{0.99}(L^+)$ bounds
    $\underline{s}_N$ (left-hand side) and $\overline{s}_N$ (right-hand side)
    computed with the ARA for \texttt{itol=0.001} based on $B=100$ replications.}
  \label{fig:ARA:VaR:0.001}
\end{figure}

\begin{figure}[htbp]
  \centering
  \includegraphics[width=\textwidth]{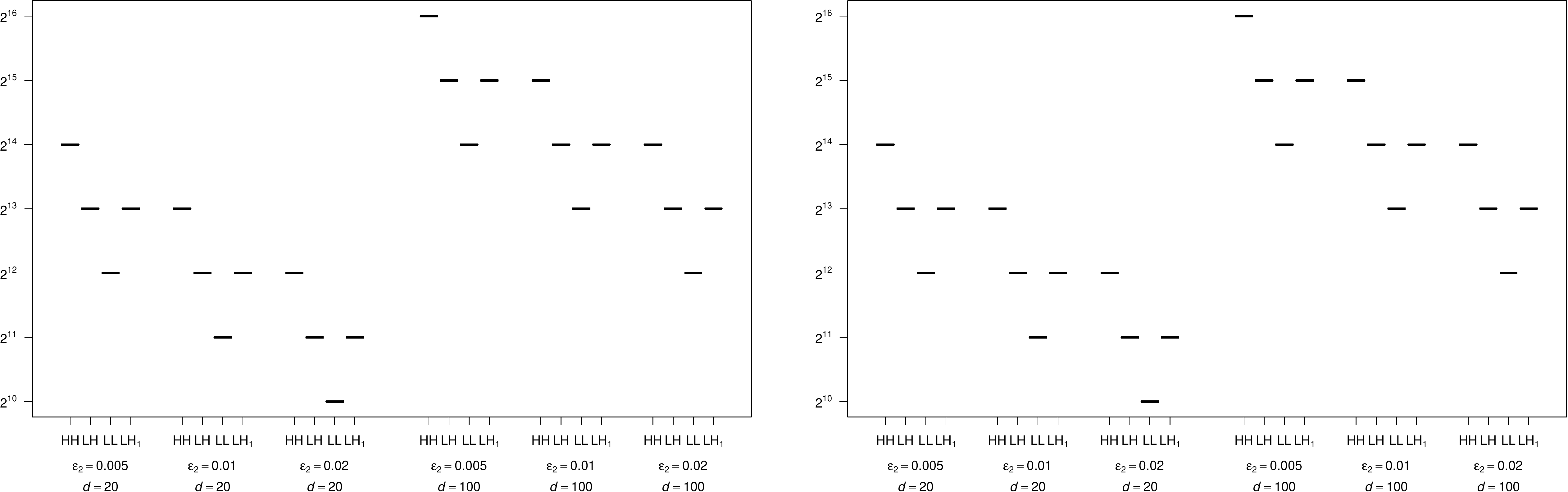}%
  \caption{Boxplots of the actual $N=2^k$ used for computing the lower and upper
    $\bVaR_{0.99}(L^+)$ bounds $\underline{s}_N$ (left-hand side) and
    $\overline{s}_N$ (right-hand side) with the ARA for \texttt{itol=0.001} based on $B=100$
    replications.}
  \label{fig:ARA:N:used:0.001}
\end{figure}

\begin{figure}[htbp]
  \centering
  \includegraphics[width=0.88\textwidth]{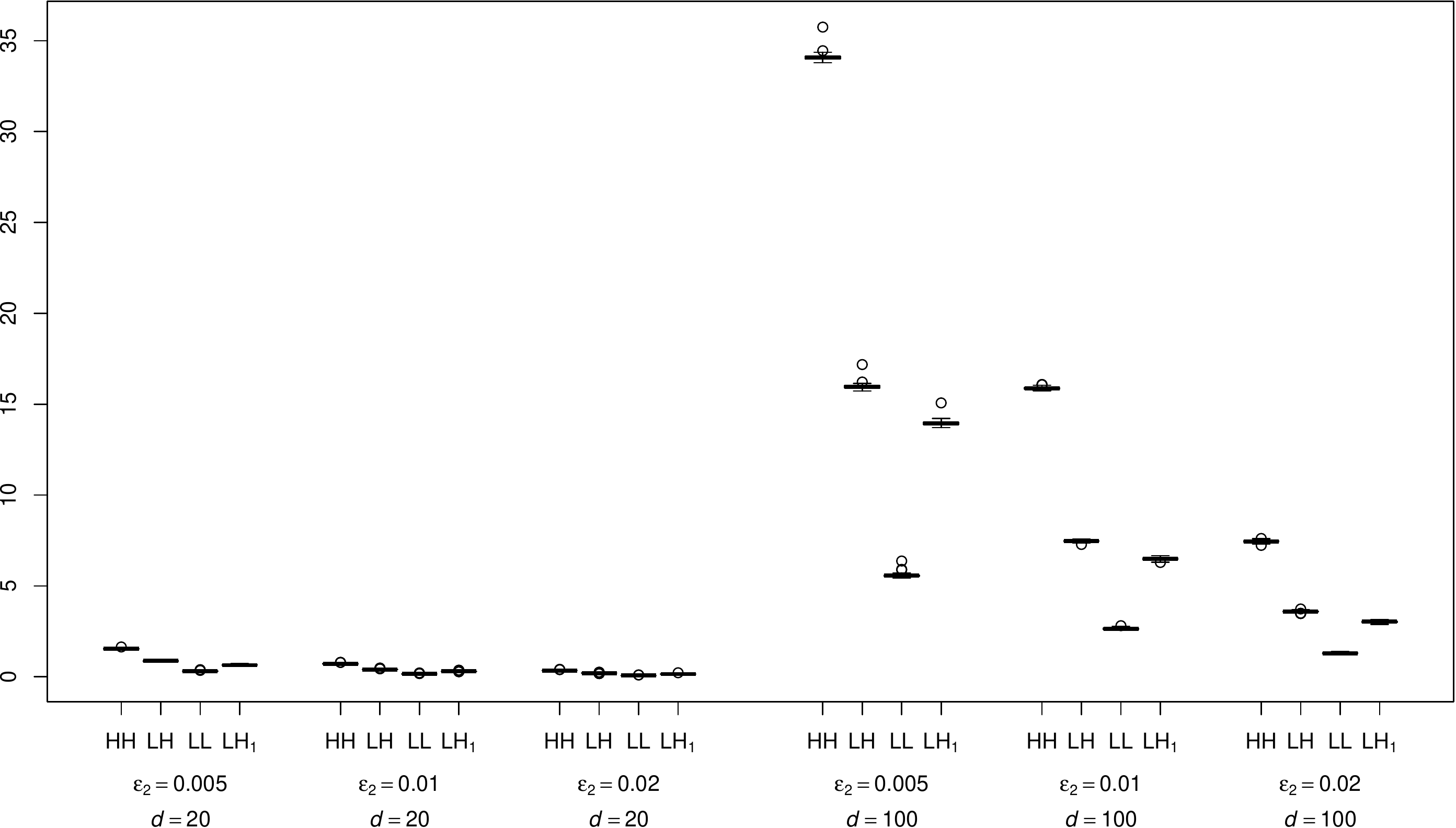}%
  \caption{Boxplots of the run time (in s) for computing the lower and upper
    $\bVaR_{0.99}(L^+)$ bounds $\underline{s}_N$ (left-hand side) and
    $\overline{s}_N$ (right-hand side) with the ARA for \texttt{itol=0.001} based on $B=100$
    replications.}
  \label{fig:ARA:runtime:0.001}
\end{figure}

\begin{figure}[htbp]
  \centering
  \includegraphics[width=\textwidth]{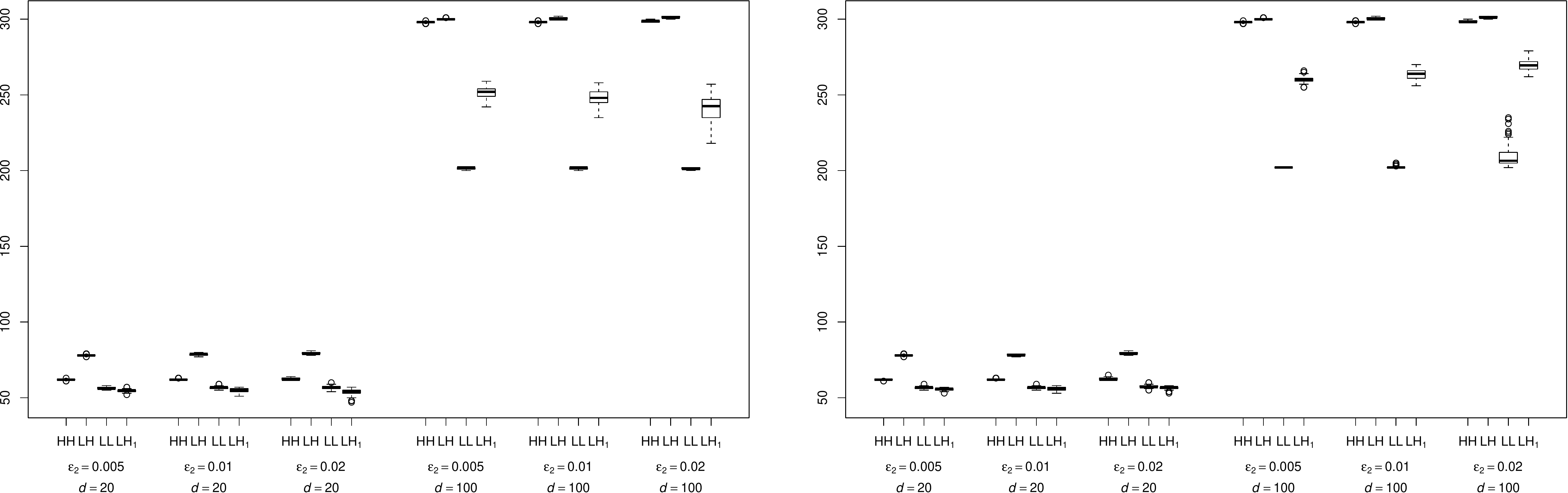}%
  \caption{Boxplots of the number of column rearrangements (measured for
    the $N$ used) for computing the lower and upper $\bVaR_{0.99}(L^+)$ bounds
    $\underline{s}_N$ (left-hand side) and $\overline{s}_N$ (right-hand side)
    with the ARA for \texttt{itol=0.001} based on $B=100$ replications.}
  \label{fig:ARA:noi:0.001}
\end{figure}

\begin{figure}[htbp]
  \centering
  \includegraphics[width=\textwidth]{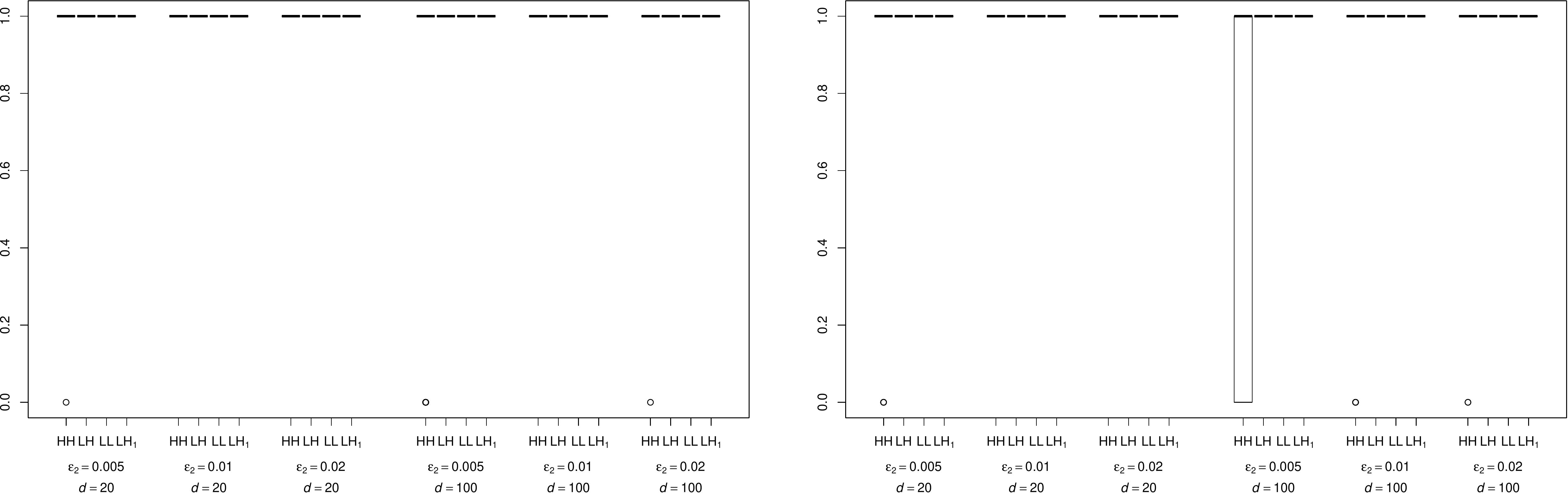}%
  \caption{Boxplots of the number of oppositely ordered columns for computing
    the lower and upper $\bVaR_{0.99}(L^+)$ bounds
    $\underline{s}_N$ (left-hand side) and $\overline{s}_N$ (right-hand side)
    with the ARA for \texttt{itol=0.001} based on $B=100$ replications.}
  \label{fig:ARA:noo:0.001}
\end{figure}

\begin{figure}[htbp]
  \begin{minipage}[t]{0.48\textwidth}
    \centering
    \includegraphics[width=\textwidth]{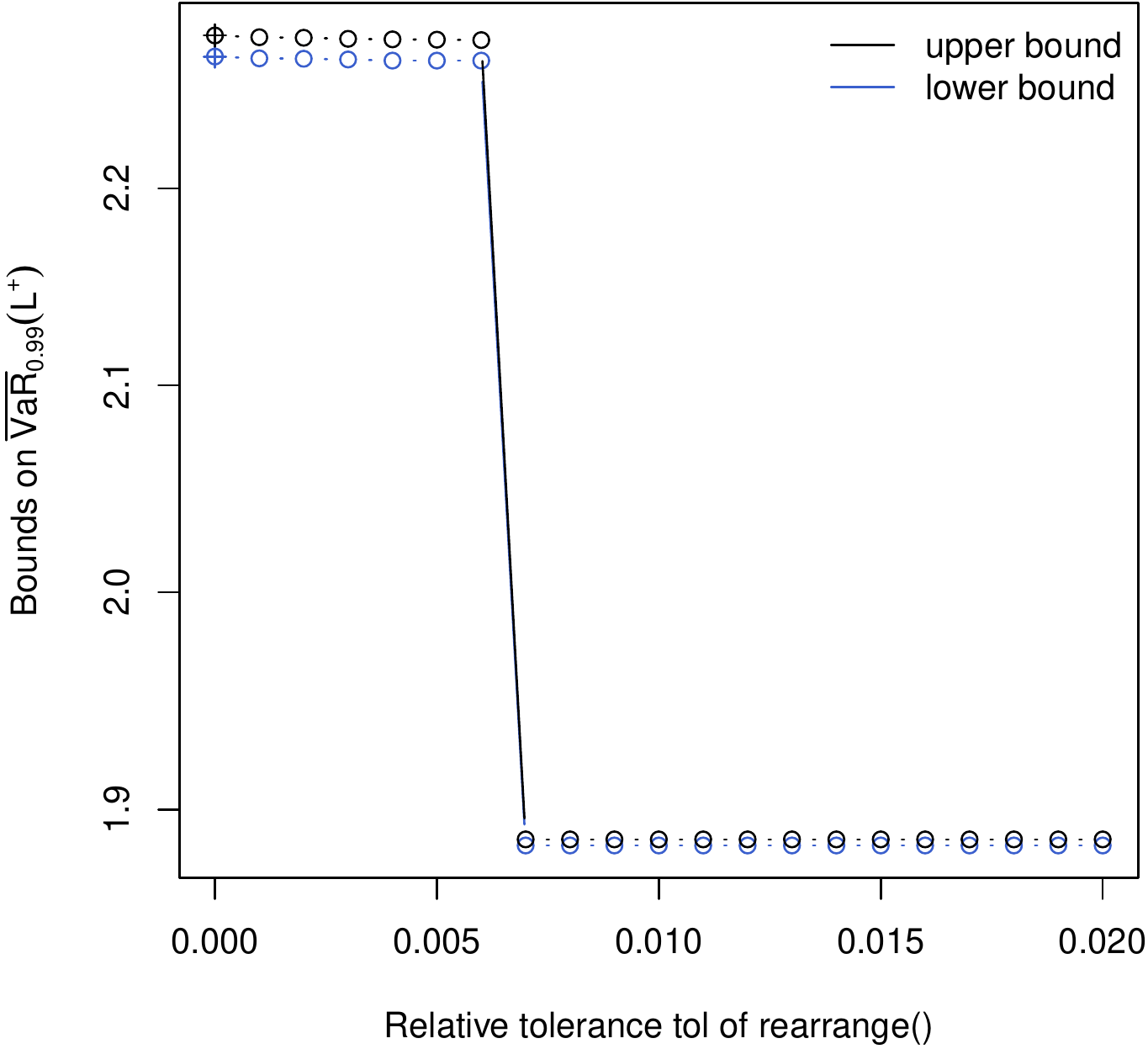}%
    \caption{Computed $\bVaR_{0.99}(L^+)$ bounds $\underline{s}_N$ and
      $\overline{s}_N$ (with \texttt{rearrange()}) depending on the chosen relative tolerance $\eps_1$
      (the cross ``$+$'' indicating the values for $\eps_1=$ \texttt{NULL})
      for an application to SMI constituents data from 2011-09-12
      to 2012-03-28; see the vignette \texttt{VaR\_bounds} for more details.}
    \label{fig:worst:VaR:bounds:application}
  \end{minipage}
  \hfill
  \begin{minipage}[t]{0.48\textwidth}
    \centering
    \includegraphics[width=\textwidth]{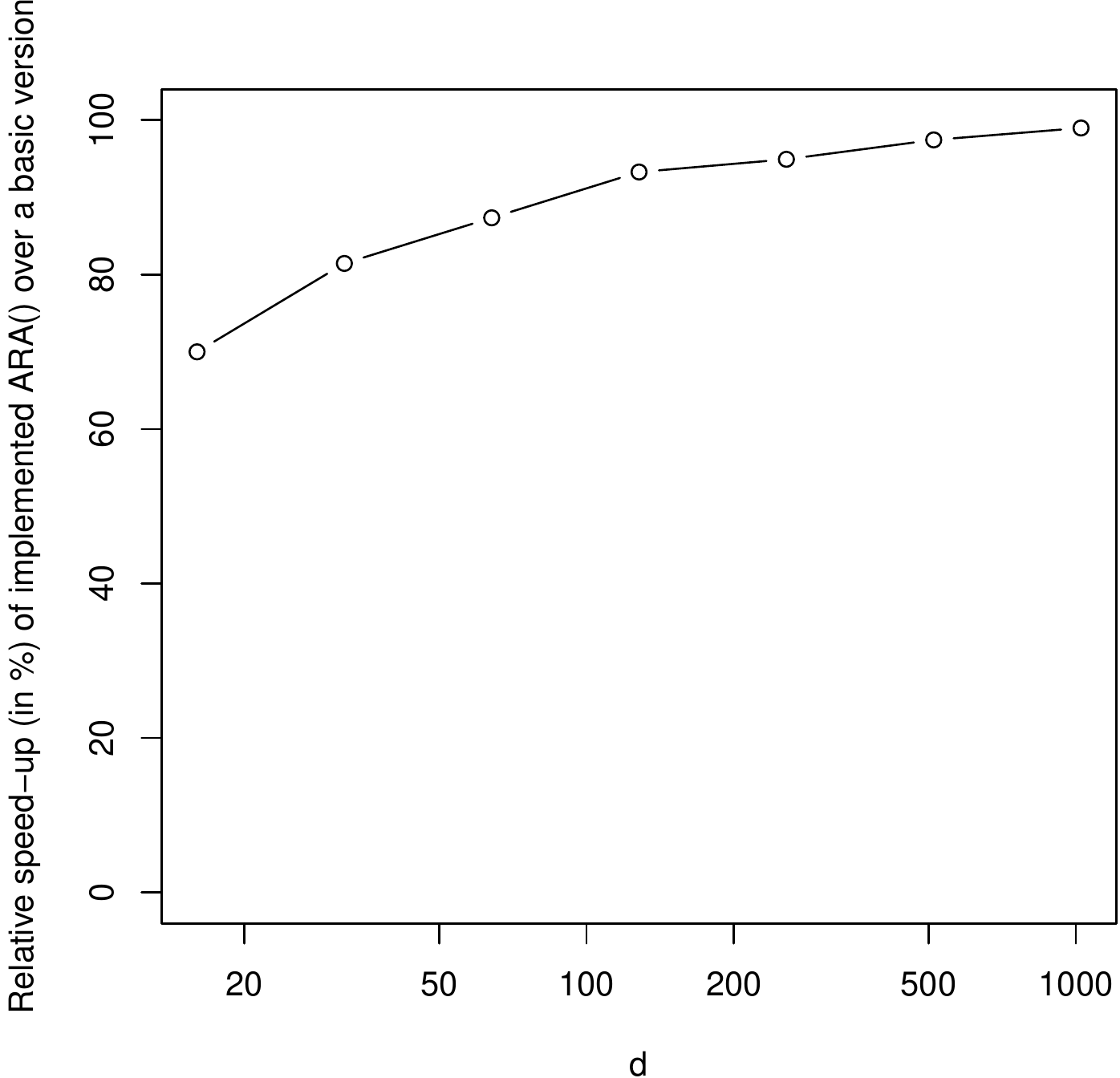}%
    \caption{Relative speed-up (in \%) of the actual implementation of
      \texttt{ARA()} over a basic version as given in the vignette
      \texttt{VaR\_bounds}.}
    \label{fig:ARA:speed:up}
  \end{minipage}
\end{figure}

\section{Conclusion}\label{sec:con}
This paper presents two major contributions to the computation of the worst
Value-at-Risk for a sum of losses with given marginals in the context of
Quantitative Risk Management.

First, we considered the homogeneous case (i.e., all margins being equal) and
addressed the dual bound approach based on
\cite[Proposition~4]{embrechtspuccettirueschendorf2013} and Wang's approach
based on \cite[Proposition~1]{embrechtspuccettirueschendorfwangbeleraj2014} for
computing worst Value-at-Risk. Although both of these approaches are
mathematically ``explicit'', care has to be exercised when computing worst
Value-at-Risk bounds with these algorithms in practice. We identified and
overcame several numerical and computational hurdles in their implementation and
addressed them using the \R\ package \texttt{qrmtools} including the vignette
\texttt{VaR\_bounds}. We covered several numerical steps such as how to compute
initial intervals for the root-finding procedure involved; a particular
example which highlights the \emph{numerical challenges} when computing worst
Value-at-Risk in general is the case of equal Pareto margins (for which we also
showed uniqueness of the root even in the infinite-mean case).

Overall, the reader should keep in mind that there is a substantial difference
between implementing a specific model (say, with $\Par(2)$ margins) where
initial intervals can be guessed or chosen ``sufficiently large'' and the proper
implementation of a result such as
\cite[Proposition~4]{embrechtspuccettirueschendorf2013} in the form of a
black-box algorithm; see the source code of \texttt{qrmtools} for more details
and the work required to go in this direction.

Second, we considered the general, i.e., inhomogeneous case. We first
investigated the Rearrangement Algorithm presented by
\cite{embrechtspuccettirueschendorf2013}. Due to its simplicity, this algorithm
by now has been widely adopted by the industry (see also
\url{https://sites.google.com/site/rearrangementalgorithm/}). Nevertheless, the
original algorithm leaves questions unanswered concerning the concrete choice of
two of its tuning parameters. These parameters were shown to have a substantial
impact on the algorithm's performance and thus need to be chosen with care. We
therefore presented an improved version of the Rearrangement Algorithm termed
Adaptive Rearrangement Algorithm. The latter improves the former in that it
addresses the aforementioned two tuning parameters and improves on the
underlying algorithmic design. The number of discretization points is chosen
automatically in an adaptive way (hence the name of the algorithm). The absolute
convergence tolerance is replaced by two relative convergence tolerances. Since
they are relative tolerances, their choice is much more intuitive. Furthermore,
the first relative tolerance is used to determine the individual convergence of
each of the lower bound $\underline{s}_N$ and the upper bound $\overline{s}_N$
for worst Value-at-Risk and the second relative tolerance is used to control how
far apart $\underline{s}_N$ and $\overline{s}_N$ are; the original version of
the Rearrangement Algorithm does not allow for such a control. The Adaptive
Rearrangement algorithm has been implemented in the \R\ package
\texttt{qrmtools}, together with conservative defaults. The implementation
contains several other improvements as well (e.g., fast accessing of columns via
lists; avoiding having to compute the row sums over all but the current column;
an extended tracing feature).

There are still several interesting questions left to be investigated. First of
all, as for the Rearrangement Algorithm, the theoretical convergence properties
of the Adaptive Rearrangement Algorithm remain an open problem. Also, as
mentioned above, it remains unclear how the rows and columns of the input
matrices for the RA or ARA can be set up in an initialization such that the run
time is minimal. Another interesting question is by how much we can reduce run
time when using the rearranged matrix from case $N=2^{k-1}$ to construct the
matrix for the case $N=2^k$ (which is why we used powers of 2 here); it remains
an open question whether the overhead of building that matrix outperforms (the
additional) sorting.

\subsection*{Acknowledgments}
We would like to thank Giovanni Puccetti (University of Milano), and, in
particular, Ruodu Wang (University of Waterloo) for valuable comments on an early draft of
this paper. Furthermore, we would like to thank Kurt Hornik (Vienna
University of Economics and Business) for pointing out algorithmic improvements.

\printbibliography[heading=bibintoc]

\appendix
\section{Appendix}

\begin{proof}[Proof of Lemma~\ref{lem:crude}]
  Consider the lower bound for $\VaR_\alpha(L^+)$. By De Morgan's Law and Boole's inequality,
  the distribution function $F_{L^+}$ of $L^+$ satisfies
  \begin{align*}
    F_{L^+}(x)&=\IP\Bigl(\,\sum_{j=1}^dL_j\le x\Bigr)\le\IP(\min_j L_j\le
    x/d)=\IP\Bigl(\bigcup_{j=1}^d\{L_j\le x/d\}\Bigr)\le\sum_{j=1}^d\IP(L_j\le
    x/d)\\
    &\le d\max_j F_j(x/d).
  \end{align*}
  Now $d\max_j F_j(x/d)\le\alpha$ if and only if $x\le d\min_j F_j^-(\alpha/d)$
  and thus $\VaR_\alpha(L^+)=F_{L^+}^-(\alpha)\ge d\min_j F_j^-(\alpha/d)$.

  Similarly, for the upper bound for $\VaR_\alpha(L^+)$, we have that
  \begin{align*}
    F_{L^+}(x)&=\IP\Bigl(\,\sum_{j=1}^dL_j\le x\Bigr)\ge\IP(\max_j L_j\le
    x/d)=\IP(L_1\le x/d,\dots,L_d\le x/d)\\
    &=1-\IP\Bigl(\bigcup_{j=1}^d\{L_j>x/d\}\Bigr)\ge\max\Bigl\{1-\sum_{j=1}^d\IP(L_j>x/d),0\Bigr\}\\
    &=\max\Bigl\{\,\sum_{j=1}^dF_j(x/d)-d+1,0\Bigr\}\ge\max\{d\min_jF_j(x/d)-d+1,0\}.
  \end{align*}
  Now $d\min_jF_j(x/d)-d+1\ge\alpha$ if and only if $x\ge d\max_j F_j^-((d-1+\alpha)/d)$
  and thus $\VaR_\alpha(L^+)=F_{L^+}^-(\alpha)\le d\max_j F_j^-((d-1+\alpha)/d)$.
\end{proof}

\begin{proof}[Proof of Proposition~\ref{prop:D}]\mbox{}
  \begin{enumerate}
  \item Let $s \ge s'$ and $t'\in[0,\frac{s'}{d}]$ such that $D(s',t')=D(s')$. Define
  \begin{align*}
      t=\frac{s-(s'-t'd)}{d}=\frac{s-s'}{d}+t'
    \end{align*}
    so that $0\le t' \le t \le \frac{s}{d}$.
    Let $\kappa = s'-t'd=s -td$. If $\kappa > 0$,
    noting that $\bar{F}$ is decreasing and that $t'\le t$, we obtain
    \begin{align*}
      D(s',t') - D(s,t)
      =\frac{d}{\kappa}\Bigl(\,\int_{t'}^{t'+\kappa}\bar{F}(x)dx-\int_{t}^{t+\kappa}\bar{F}(x)dx\Bigr)\ge 0,
    \end{align*}
    so that $D(s) \le D(s,t) \le D(s',t')=D(s')$. If $\kappa=0$ then $D(s')=D(s',\tfrac{s'}{d}) =d\bar F(\tfrac{s'}{d})
    \ge d \bar F(\tfrac{s}{d}) \ge D(s)$.
  \item Recall that $D(s,t) = \frac{d}{s-td} \int_{t}^{t+(s-td)} \bar{F}(x)\,dx$. Using the transformation $z=(x-t)/(s-td)$, we have
    \begin{align*}
      D(s,t) = d \int_0^1 \bar{F} (s z+t(1-z d))\,dz
    \end{align*}
    Define $C=\{(s,t)\,|\,0 \le s < \infty,\ 0\le t \le \frac{s}{d}\}$, and note that $C$ is convex. Furthermore,
    if $\bar{F}$ is convex, then $D(s,t)$ is jointly convex in $s$ and $t$ on $C$ since for $\lambda \in (0,1)$,
    \begin{align*}
      &\phantom{{}={}}D(\lambda s_1+(1-\lambda)s_2, \lambda t_1+(1-\lambda)t_2)\\
      &= d \int_0^1 \bar{F}((\lambda s_1+(1-\lambda)s_2)z+ (\lambda t_1+(1-\lambda)t_2)(1 -z d))\,dz\\
      &= d \int_0^1 \bar{F} (\lambda(s_1z+t_1(1-z d))+(1-\lambda)(s_2z+t_2(1-z d)))\,dz \\
      &\le \int_0^1 \lambda \bar{F} (((s_1+t_1(1-z d)))+(1-\lambda)\bar{F}((s_2+t_2(1-z d))))\,dz\\
      &= \lambda D(s_1,t_1) +(1-\lambda) D(s_2,t_2)\qedhere
    \end{align*}
  \end{enumerate}
\end{proof}

\begin{proof}[Proof of Proposition~\ref{prop:unique:root:Par}]
  First consider $\theta\neq1$. Using \eqref{eq:Ibar:Par}, one can
  rewrite $h(c)$ as
  \begin{align*}
    h(c)=\frac{c^{-1/\theta+1}\frac{\theta}{1-\theta}(1-(\frac{1-\alpha}{c}-(d-1))^{-1/\theta+1})}{1-\alpha-dc}-\frac{(d-1)(\frac{1-\alpha}{c}-(d-1)^{-1/\theta}+1)}{c^{1/\theta}d}.
  \end{align*}
  Multiplying with $c^{1/\theta}d$ and rewriting the expression, one sees that
  $h(c)=0$ is equivalent to $h_1(x_c)=0$ where
  \begin{align}
    x_c=\frac{1-\alpha}{c}-(d-1)\label{xc}
  \end{align}
  (which is in $(1,\infty)$ for $c\in(0,(1-\alpha)/d)$) and
  $h_1(x)=d\frac{\theta}{1-\theta}\frac{1-x^{-1/\theta+1}}{x-1}-((d-1)x^{-1/\theta}+1)$.
  It is easy to see that $h_1(x)=0$ if and only if
  $h_2(x)=0$, where
  \begin{align}
    h_2(x)=(d/(1-\theta)-1)x^{-1/\theta+1}-(d-1)x^{-1/\theta}+x-(d\theta/(1-\theta)+1),\quad x\in(1,\infty).\label{eq:def:h2}
  \end{align}
  We are done for $\theta\neq1$ if we show that $h_2$ has a unique root on
  $(1,\infty)$. To this end, note that $h_2(1)=0$ and
  $\lim_{x\uparrow\infty}h_2(x)=\infty$. Furthermore,
  \begin{align*}
    h_2'(x)&=(1-1/\theta)(d/(1-\theta)-1)x^{-1/\theta}+(d-1)\theta
             x^{-1/\theta-1}+1,\\
    h_2''(x)&=(d+\theta-1)/\theta^2 x^{-1/\theta-1} - (1/\theta+1)(d-1)/\theta x^{-1/\theta-2}.
  \end{align*}
  It is not difficult to check that $h_2''(x)=0$ if and only if
  $x=\frac{(d-1)(1+\theta)}{d+\theta-1}$ (which is greater than 1 for
  $d>2$). Hence, $h_2$ can have at most one root. We are done if we find an
  $x_0\in(1,\infty)$ such that $h_2(x_0)<0$, but this is guaranteed by the fact that
  $\lim_{x\downarrow 1}h_2'(x)=0$ and $\lim_{x\downarrow
    1}h_2''(x)=-(d-2)/\theta<0$ for $d>2$.

  The proof for $\theta=1$ works similarly; in this case, $h_2$ 
  is given by
  \begin{align}
    h_2(x)=x^2+x(-d\log(x)+d-2)-(d-1),\quad x\in(1,\infty),\label{eq:def:h2:th1}
  \end{align}
  and the unique point of inflection of $h_2$ is $x=d/2$.
\end{proof}

\begin{proof}[Proof of Proposition~\ref{prop:cl:cu:Par}]
  First consider $c_l$ and $\theta\neq1$. Instead of $h$,
  \eqref{xc} and \eqref{eq:def:h2} allow us to study
  \begin{align*}
    h_2(x)=(d/(1-\theta)-1)x^{-1/\theta+1}-(d-1)x^{-1/\theta}+x-(d\theta/(1-\theta)+1),\quad x\in[1,\infty).
  \end{align*}
  Consider the two cases $\theta\in(0,1)$ and $\theta\in(1,\infty)$
  separately. If $\theta\in(0,1)$, then $d/(1-\theta)-1>d-1\ge0$ and
  $x^{-1/\theta+1}\ge x^{-1/\theta}$ for all $x\ge1$, so
  $h_2(x)\ge
  ((d/(1-\theta)-1)-(d-1))x^{-1/\theta}+x-(d\theta/(1-\theta)+1)\ge
  x-(d\theta/(1-\theta)+1)$
  which is 0 if and only if $x=d\theta/(1-\theta)+1$. Setting this equal to
  $x_c$ (defined in \eqref{xc}) and solving for $c$ leads to the $c_l$ as
  provided. If $\theta\in(1,\infty)$, then using $x^{-1/\theta}\le 1$ leads to
  $h_2(x)\ge(d/(1-\theta)-1)x^{-1/\theta+1}+x$ which is 0 for $x\ge1$ if and
  only if $x=(d/(\theta-1)+1)^\theta$. Setting this equal to $x_c$ and solving
  for $c$ leads to the $c_l$ as provided.

  Now consider $\theta=1$. Similar as before, we can consider
  \eqref{xc} and \eqref{eq:def:h2:th1}.
  By using that $\log x\le x^{1/e}$ and $x\ge -x^{1+1/e}$ for $x\in[1,\infty)$,
  we obtain
  $h_2(x)\ge x^2+x(-dx^{1/e}+d-2)-(d-1)\ge x^2-(d+1)x^{1+1/e}$
  which is 0 if and only if $x=(d+1)^{e/(e-1)}$. Setting this equal to $x_c$ and
  solving for $c$ leads to the $c_l$ as provided.

  Now consider $c_u$. It is easy to see that the inflection point of $h_2$
  provides a lower bound $x_c$ on the root of $h_2$. As derived in the proof of Proposition~\ref{prop:unique:root:Par},
  the point of inflection is $x=x_c:=\frac{(d-1)(1+\theta)}{d+\theta-1}$ for
  $\theta\neq1$ and $x=d/2$ for $\theta=1$. Solving
  $x_c=(1-\alpha)/c-(d-1)$ for $c$ then leads to $c_u$ as stated.
\end{proof}

\bigskip
\noindent{\bfseries Marius Hofert}\\
Department of Statistics and Actuarial Science\\
University of Waterloo\\
200 University Avenue West, Waterloo, ON, N2L 3G1\\
\href{mailto:marius.hofert@uwaterloo.ca}{\nolinkurl{marius.hofert@uwaterloo.ca}}

\bigskip
\noindent{\bfseries Amir Memartoluie}\\
Cheriton School of Computer Science\\
University of Waterloo\\
200 University Avenue West, Waterloo, ON, N2L 3G1\\
\href{mailto:amir.memartoluie@uwaterloo.ca}{\nolinkurl{amir.memartoluie@uwaterloo.ca}}

\bigskip
\noindent{\bfseries David Saunders}\\
Department of Statistics and Actuarial Science\\
University of Waterloo\\
200 University Avenue West, Waterloo, ON, N2L 3G1\\
\href{mailto:david.saunders@uwaterloo.ca}{\nolinkurl{david.saunders@uwaterloo.ca}}

\bigskip
\noindent{\bfseries Tony Wirjanto}\\
Department of Statistics and Actuarial Science\\
University of Waterloo\\
200 University Avenue West, Waterloo, ON, N2L 3G1\\
\href{mailto:tony.wirjanto@uwaterloo.ca}{\nolinkurl{tony.wirjanto@uwaterloo.ca}}

\end{document}
